\documentclass[12pt]{article} 

\usepackage{amssymb,amsmath}
\usepackage{bbm}%
\usepackage{euscript}
\usepackage{balance}
\usepackage{multirow}%
\usepackage[cm]{fullpage}%

\usepackage{hyperref}%
\usepackage[amsmath,thmmarks]{ntheorem}%
\theoremseparator{.}%

\hypersetup{%
   breaklinks,%
   ocgcolorlinks, colorlinks=true,%
   allcolors=[rgb]{0.2,0.0,0.0}%
}

\usepackage{enumerate}
\usepackage{graphicx}

\usepackage{color}
\usepackage{xspace}
\usepackage{paralist}%
\usepackage{picins}%

\newcommand{\Set}[2]{\left\{ #1 \;\middle\vert\; #2 \right\}}
\newcommand{\myqedsymbol}{\rule{2mm}{2mm}}

\newcommand{\Circle}{\sigma}

\definecolor{blue25}{rgb}{0.0,0,0.1}

\newcommand{\emphi}{\textit}

\newcommand{\A}{\EuScript{A}}

\renewcommand{\Re}{{\rm I\!\hspace{-0.025em} R}}
\newcommand{\brc}[1]{\left\{ {#1} \right\}}
\newcommand{\pth}[2][\!]{#1\left({#2}\right)}

\newcommand{\query}{q}

\renewcommand{\Re}{\mathbb{R}}

\newcommand{\barX}[1]{\overline{#1}}
\newcommand{\bp}{\barX{p}}
\newcommand{\wpi}{\widehat{\pi}}

\newcommand{\D}{\EuScript{D}}

\newcommand {\mm}[1] {\ifmmode{#1}\else{\mbox{\(#1\)}}\fi}
\newcommand{\dir}[1]       {\mm{\rm d}{#1}}

\def\P{\EuScript{P}}
\def\T{\EuScript{T}}

\def\T{\EuScript{T}}

\newcommand{\reals}{\mathbb{R}}

\newcommand{\atgen}{\symbol{'100}}

\def\eps{\varepsilon}

\def\dist{d}

\def\Sup{\mathop{\mathrm{Sup}}}

\newcommand{\PVD}{\VorChar_{\Pr}}

\newcommand{\VorChar}{\EuScript{V}}
\newcommand{\NZVD}{\ensuremath{\VorChar_{\neq 0}}\xspace}

\newcommand{\NZVorCell}{\ensuremath{\mathrm{cell}_{\ne 0}}\xspace}

\newcommand{\Term}[1]{\textsf{#1}}

\newcommand{\NN}{\Term{NN}\xspace}

\newcommand{\spread}{\lambda}

\newcommand{\Vor}{\Term{Vor}\xspace}

\newcommand{\PNN}{\Term{PNN}\xspace}

\newcommand{\NZNN}{\ensuremath{\NN_{\ne 0}}\xspace}

\renewcommand{\th}{th\xspace}

\newcommand{\pdf}{\textsf{pdf}\xspace}
\newcommand{\cdf}{\textsf{cdf}\xspace}

\def\polylog{\mathop{\mathrm{polylog}}}
\def\eps{\varepsilon}

\def\Int{\mathop{\mathrm{int}}}

\def\index{\text{data structure}}
\def\anindex{\text{a \index}}

\def\Pr{\mathrm{Pr}}
\def\E{\mathsf{E}}

\newcommand{\remove}[1]{}
\newcommand{\si}[1]{#1}

\newcommand{\disk}{D}
\newcommand{\diskA}{E}

\newcommand{\diskW}{W}

\newcommand{\HLinkShort}[2]{\hyperref[#2]{#1\ref*{#2}}}
\newcommand{\HLink}[2]{\hyperref[#2]{#1~\ref*{#2}}}
\newcommand{\HLinkPage}[2]{\hyperref[#2]{#1~\ref*{#2}%
      $_\text{p\pageref{#2}}$}}
\newcommand{\HLinkPageOnly}[1]{\hyperref[#1]{Page~\refpage*{#1}%
      $_\text{p\pageref{#1}}$}}
\newcommand{\HLinkSuffix}[3]{\hyperref[#2]{#1\ref*{#2}{#3}}}
\newcommand{\HLinkPageSuffix}[3]{\hyperref[#2]{#1\ref*{#2}%
      #3$_\text{p\pageref{#2}}$}}%

\providecommand{\eqlab}[1]{}%
\renewcommand{\eqlab}[1]{\label{equation:#1}}
\newcommand{\Eqref}[1]{\HLinkSuffix{Eq.~(}{equation:#1}{)}}

\newcommand{\figlab}[1]{\label{fig:#1}}
\newcommand{\figref}[1]{\HLink{Figure}{fig:#1}}

\newcommand{\lemlab}[1]{\label{lemma:#1}}
\newcommand{\lemref}[1]{\HLink{Lemma}{lemma:#1}}%

\newcommand{\thmlab}[1]{{\label{theo:#1}}}
\newcommand{\thmref}[1]{\HLink{Theorem}{theo:#1}}

\newcommand{\seclab}[1]{\label{sec:#1}}
\newcommand{\secref}[1]{\HLink{Section}{sec:#1}}

\newcommand{\corlab}[1]{\label{cor:#1}}
\newcommand{\corref}[1]{\HLink{Corollary}{cor:#1}}

\newcommand{\DiskSet}{\EuScript{D}}%

\newcommand{\IntervalChar}{\EuScript{I}}%
\newcommand{\IntervalX}[1]{\IntervalChar\pth{#1}}%
\newcommand{\radiusX}[1]{\mathsf{r}\pth{#1}}%

\newcommand{\cardin}[1]{\left| {#1} \right|}

\newcommand{\constA}{\xi}
\newcommand{\ceil}[1]{\left\lceil {#1} \right\rceil}

\newcommand{\Wtau}{\diskW_{\ominus \tau}}
\def\qq{\bar q}
\def\ranX{\mathrm{X}}
\def\QQ{\bar Q}

\newcommand{\PankajThanks}[1]{\thanks{
      Department of Computer Science; %
      Duke University; %
      Durham, NC, 27708, USA; %
      {\tt pankaj\atgen{}cs.duke.edu}; %
      {\tt \url{http://www.cs.duke.edu/\string~pankaj/}}. %
      {#1}
   }%
}
\newcommand{\BorisThanks}[1]{%
   \thanks{%
      Polytechnic School of Engineering, New York University; %
      New-York City, NY, USA; %
      {#1}
   }
}

\newcommand{\JeffThanks}[1]{%
   \thanks{%
      School of Computing; %
      The University of Utah;
      {#1}
   }
}

\newcommand{\KeThanks}[1]{%
   \thanks{%
      Hong Kong University of Science and Technology;
      {#1}
   }
}

\newcommand{\WuzhouThanks}[1]{%
   \thanks{%
      Apple Inc.
      {#1}
   }
}

\newcommand{\SarielThanks}[1]{\thanks{Department of Computer Science;
      University of Illinois; 201 N. Goodwin Avenue; Urbana, IL,
      61801, USA; {\tt sariel\atgen{}illinois.edu}; {\tt
         \url{http://sarielhp.org/}.} #1}}

\theoremstyle{plain}%
\newtheorem{theorem}{Theorem}[section]

\newtheorem{lemma}[theorem]{Lemma}

\newtheorem{corollary}[theorem]{Corollary}

\theoremstyle{plain}%
\theoremheaderfont{\sf} \theorembodyfont{\upshape}%
\newtheorem*{remark:unnumbered}[theorem]{Remark}%
\theoremheaderfont{\em}%
\theorembodyfont{\upshape}%
\theoremstyle{nonumberplain}%
\theoremseparator{}%
\theoremsymbol{\myqedsymbol}%
\newtheorem{proof}{Proof:}%

\newcommand{\IncludeGraphics}[2][]{%
   \typeout{}%
   \typeout{Graphics: #2}%
   \typeout{\ includegraphics[#1]{#2}}%
   \includegraphics[#1]{#2}
   \typeout{}%
}

\begin{document}

\title{Nearest-Neighbor Searching Under Uncertainty II%
   \footnote{%
      A preliminary version of this article appeared in
      \textit{Proceedings of the ACM Symposium on Principles of
         Database Systems} (PODS), 2013. The title
      ``\textit{Nearest-Neighbor Searching Under Uncertainty I}'' has
      been reserved for the journal version of
      ~\cite{aesz-nnsuu-12}. %
      Most of the work on this paper was done while W. Zhang was at
      Duke University.%
   }%
}%
   
\author{%
   Pankaj~K.~Agarwal%
   \PankajThanks{Research on this paper was partially supported by NSF
      under grants CCF-09-40671, CCF-10-12254, CCF-11-61359, and
      IIS-14-08846.}%
   \and%
   Boris Aronov%
   \BorisThanks{Research on this paper has been partially supported by
      NSF grants CCF-08-30691, CCF-11-17336, and CCF-12-18791, and by
      \si{NSA MSP} Grant H98230-10-1-0210.%
   }%
   \and%
   Sariel Har-Peled%
   \SarielThanks{Research on this paper was partially supported by NSF
      grants CCF-09-15984 and CCF-12-17462.}%
   \and%
   Jeff M.  Philips%
   \JeffThanks{The research on this paper was partially supported by
      NSF CCF-1350888, IIS-1251019, and ACI-1443046. }%
   \and%
   Ke Yi%
   \KeThanks{%
      K.~Yi is supported by \si{HKRGC} under grants \si{GRF}-621413
      and \si{GRF}-16211614.  }%
   \and%
   Wuzhou Zhang%
   \WuzhouThanks{Research on this paper was partially supported by NSF
      under grants CCF-09-40671, CCF-10-12254, CCF-11-61359, and
      IIS-14-08846.}%
}%

\maketitle
\begin{abstract}
    Nearest-neighbor search, which returns the nearest neighbor of a
    query point in a set of points, is an important and widely studied
    problem in many fields, and it has wide range of applications.  In
    many of them, such as sensor databases, location-based services,
    face recognition, and mobile data, the location of data is
    imprecise. We therefore study nearest-neighbor queries in a
    probabilistic framework in which the location of each input point
    is specified as a probability distribution function.  We present
    efficient algorithms for
    \begin{inparaenum}[(i)]
        \item computing all points that are nearest neighbors of a
        query point with nonzero probability; and
        \item estimating the probability of a point being the nearest
        neighbor of a query point, either exactly or within a
        specified additive error.
    \end{inparaenum}
\end{abstract}

\section{Introduction}
Nearest-neighbor search is a fundamental problem in data management.
It has applications in such diverse areas as spatial databases,
information retrieval, data mining, pattern recognition, etc.  In its
simplest form, it asks for preprocessing a set $S$ of $n$ points in
$\Re^d$ into $\anindex$ so that given any query point $q$, the nearest
neighbor ($\NN$) of~$q$ in~$S$ can be reported efficiently.  This
problem has been studied extensively in database, machine learning,
and computational geometry communities, and is now relatively well
understood.  However, in some of the applications mentioned above,
data are imprecise and are often modeled as probabilistic
distributions.  This has led to a flurry of research activities on
query processing over probabilistic data, including the \NN problem;
see~\cite{a-mmud-09,drs-pddd-09} for surveys on uncertain data, and
see, e.g.,~\cite{cxycs-uvdvd-10,ls-aiap-07} for application scenarios
of \NN search under uncertainty.

Despite many efforts devoted to the probabilistic \NN problem, it
still lacks a theoretical foundation. Specifically, not only are we
yet to understand its complexity (is the problem inherently more
difficult than on precise data?), but we also lack efficient
algorithms to solve it. Furthermore, existing solutions all use
heuristics without nontrivial performance guarantees.  This paper
addresses some of these issues.

\subsection{Problem definition}

An \emphi{uncertain} point $P$ in $\Re^2$ is represented as a
continuous probability distribution defined by a probability density
function (\textsf{p{d}f}) $f_P\colon \Re^2 \rightarrow \Re_{\geq 0}$;
$f_P$ may be a parametric \pdf such as a uniform distribution or a
Gaussian distribution, or may be a non-parametric \pdf such as a
histogram\footnote{If the location of data is precise, we call it
   \emph{certain}.  The probabilistic model we use is often called the
   \emphi{locational model}, where the location of an uncertain point
   follows the given distribution.  This is to be contrasted with the
   \emphi{existential model}, where each point has a precise location
   but it appears with a given probability.}. %
The \emphi{uncertainty region} of~$P$ (or the \emphi{support}
of~$f_P$) is the set of points for which $f_P$ is positive, i.e.,
$\Sup$ $f_P = \brc{x \in \Re^2 \mid f_P(x) > 0}$.  We assume $P$ has a
bounded uncertainty region: if $f_P$ is Gaussian, we work with the
truncated Gaussian, as in~\cite{bsi-estkp-08,ccmc-pvecn-08}. We also
consider the case where $P$ is represented as a discrete distribution
defined by a finite set $P = \brc{ p_1, \ldots, p_k} \subset \Re^2$
along with a set of \emph{location probabilities}
$\{w_1, \ldots, w_k\} \subset (0, 1]$, where
$w_i = \Pr[P \text{ is } p_i]$ and $\sum_{i = 1}^k w_i = 1$; and we
say that $P$ has a discrete distribution of \textit{description
   complexity} $k$.  Let $\P = \{P_1, \ldots, P_n\}$ be a set of $n$
uncertain points in~$\Re^2$, and let $\dist(\cdot, \cdot)$ be the
Euclidean distance.

Fix a point $q \in \Re^2$ and an integer $i \in \{1,\ldots, n\}$.  We
define $\pi_i(q) = \pi(P_i, q)$ to be the probability of $P_i\in \P$
being the nearest neighbor of $q$, referred to as the
\emphi{quantification probability} of $q$ (for $P_i$).  Next, let
$g_{q, i}$ be the \pdf of the distance between $q$ and $P_i$. That is,
\begin{align*}
  g_{q, i}(r) = \Pr[ r \leq \dist(q, P_i) \leq r + \dir r ]  / \dir r.
\end{align*}
See \figref{1} for an example of $g_{q, i}$. Let
$G_{q,i}(r) = \int_0^r g_{q,i}(r') \dir r'$ denote the cumulative
distribution function (\cdf) of the distance between $q$ and
$P_i$. Note that if $P_i$ is the NN of $q$ and $\dist(P_i,q)=r$ then
$\dist(P_j, q) > r$ for all $j \ne i$. Therefore $\pi_i(q)$ can be
expressed as follows:
\begin{align}
  \eqlab{definition} \pi_i(q) = \int_0^{\infty} g_{q, i}(r)
  \prod_{j\neq i} \bigl(1 - G_{q,j}(r)\bigr) \dir r.
\end{align}

If $P_i$'s are represented by discrete distributions, then
\Eqref{definition} can be rewritten as follows:
\begin{align}
  \eqlab{definition2} \pi_i(q) = \sum_{p_{is} \in P_i} w_{is}
  \prod_{j \neq i} \Bigl(1 - G_{q, j}\bigl(\dist(p_{is}, q)\bigr)
  \Bigr),
\end{align}
where
\begin{math}%
    \displaystyle%
    G_{q, j}(r) = \sum_{\dist(p_{jt}, q) \leq r} w_{jt}.
\end{math}

Given a set $\P$ of $n$ uncertain points, the \emphi{probabilistic
   nearest neighbor} ($\PNN$) problem is to preprocess $\P$ into
$\anindex$ so that, for any given query point $q$, we can efficiently
return all pairs $(P_i, \pi_i(q))$ with $\pi_i(q)>0$.

\begin{figure}
    \centering
    \begin{tabular}{ccc}
      \IncludeGraphics[width=0.24\textwidth]{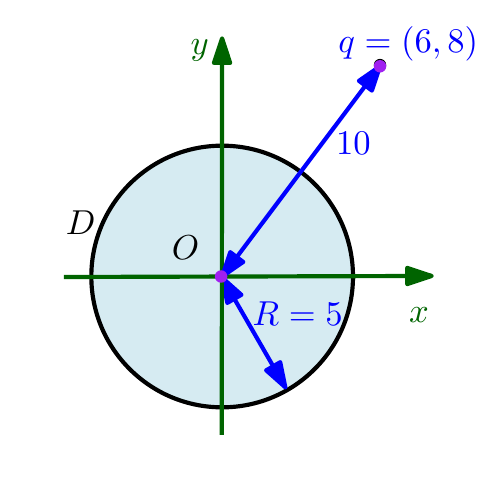}
      &\hspace*{15mm}&
                       \IncludeGraphics[width=0.39\textwidth]%
                       {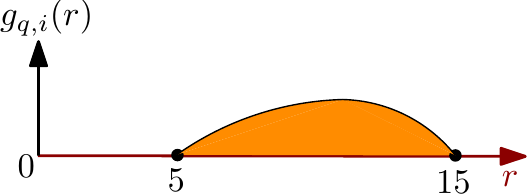} \\
      && \\
      \small (a) && \small (b)
    \end{tabular}
    \caption{(a) $P_i$ is represented by a uniform distribution
       defined on a disk $D$ of radius $R = 5$ and centered at
       origin~$O$, $q = (6, 8)$; (b) $g_{q, i}(r)$, the \pdf of the
       distance function between $q$ and $P_i$.}%
    \figlab{1}
\end{figure}

Usually, the \PNN problem is divided into the following two
subproblems, which are often considered separately.

\paragraph{Nonzero \NN{}s.} %
The first subproblem is to find all the $P_i$'s with $\pi_i(q)>0$
without computing the actual quantification probabilities, i.e., to
find
\begin{align*}
  \NZNN(q, \P) = \{P_i \mid \pi_i(q) > 0\}.
\end{align*}
If the point set $\P$ is obvious from the context, we drop the
argument $\P$ from $\NZNN(q, \P)$, and write it as $\NZNN(q)$. Note
that $\NZNN(q)$ depends (besides $q$) only on the uncertainty regions
of the uncertain points, but not on the actual \pdf's.

A possible approach to compute nearest neighbors is to use Voronoi
diagrams.  For example, the standard Voronoi diagram of a set of
certain points in $\Re^2$ is the planar subdivision so that all points
in the same face have the same nearest neighbor.  In our case, we
define the \emphi{nonzero Voronoi diagram}, denoted by $\NZVD(\P)$, to
be the subdivision of $\Re^2$ into maximal connected regions such that
$\NZNN(q)$ is the same for all points $q$ within each region. That is,
for a subset $\T \subseteq \P$, let
\begin{align}%
  \eqlab{NZVorCell}%
  \NZVorCell(\T) = \{ q \in \Re^2 \mid \NZNN(q) = \T \}.
\end{align}

Although there are $2^n$ subsets of $\P$, we will see below that only
a small number of them have nonempty Voronoi cells.  The planar
subdivision $\NZVD(\P)$ is induced by all the nonempty
$\NZVorCell(\T)$'s for $\T \subseteq \P$. The \emph{(combinatorial)
   complexity} of $\NZVD(\P)$ is the total number of vertices, edges,
and faces in $\NZVD(\P)$.  The complexity of the Voronoi diagram is
often regarded as a measure of the complexity of the corresponding
nearest-neighbor problem.

In this paper, we study the worst-case complexity of $\NZVD(\P)$ and
how it can be efficiently constructed.  In addition, once we have
$\NZVD(\P)$, it can be preprocessed into a point-location structure to
support \NZNN queries in logarithmic time.

\paragraph{Computing quantification probabilities.} %

The second subproblem is to compute the quantification probability
$\pi_i(q)$ for a given $q$ and $P_i$.  Exact values of these
probabilities are often unstable --- a far away point can affect these
probabilities --- and computing them requires complex $n$-dimensional
integration (see~\Eqref{definition}), which is often expensive. As
such, we resort to computing $\pi_i(q)$ approximately within a given
additive error tolerance $\eps\in(0, 1)$.  More precisely, we aim at
returning a value $\wpi_i(q)$ such that
$|\pi_i(q) - \wpi_i(q)| \le \eps$.

\subsection{Previous work}

\paragraph{Nonzero \NN{}s.} %

\cite{se-gvdus-08} showed that if the uncertainty regions of the
points in $\P$ are disks, then the complexity of $\NZVD(\P)$ is
$O(n^4)$ (though they did not use this term explicitly); they did not
offer any lower bound.  If one only considers those cells of
$\NZVD(\P)$ in which $\NZNN(q)$ contains only one uncertain point
$P_i$, i.e., only $P_i$ has a non-zero probability of being the $\NN$
of $q$, they showed that the complexity of these cells is $O(n)$.
Note that for such a cell, we always have $\pi_i(q)= 1$ for any~$q$ in
the cell, so they form the \emph{guaranteed Voronoi diagram}. Probably
unaware of the work by \cite{se-gvdus-08}, \cite{cxycs-uvdvd-10}
proved an exponential upper bound for the complexity of the nonzero
Voronoi diagram, which they referred to as UV-diagram.

The nonzero Voronoi diagram is not the only way to find the nonzero
\NN{}s.  \cite{ckp-qidmo-04} designed a branch-and-prune solution
based on the $R$-tree.  Recently, ~\cite{zcmrz-vbnns-13} proposed to
combine the nonzero Voronoi diagram with R-tree-like bounding
rectangles. These methods do not provide any nontrivial performance
guarantees.

\paragraph{Computing quantification probabilities.} %
Computing the quantification probabilities has attracted much
attention in the database community. \cite{ckp-qidmo-04} used
numerical integration, which is quite expensive. \cite{ccmc-pvecn-08}
and \cite{bekmr-nppas-11} proposed some filter-refinement methods to
give upper and lower bounds on the quantification
probabilities. \cite{kkr-pnnqu-07} took a random sample from the
continuous distribution of each uncertain point to convert it to a
discrete one, so that the integration becomes a sum, and they
clustered each sample to further reduce the complexity of the query
computation.  \cite{dymtv-psqeu-05} considered the problem of
reporting points $P_i$ for which $\pi_i(q)$ exceeds some given
threshold.  We note that these methods are best-effort based: they do
not always give the $\eps$-error that we aim at --- how tight the
resulting bounds are depends on the data.

\paragraph{Other variants of the problem.} %

The \PNN problem we focus on in this paper is the most commonly
studied version of the problem, but many variants and extensions have
been considered.

Besides using the quantification probability, one can also consider
the expected distance from a query point $q$ to an uncertain point,
and return the one minimizing the expected distance as the nearest
neighbor; this was studied by \cite{aesz-nnsuu-12}. This \NN
definition is easier since the expected distance to each uncertain
point can be computed separately, whereas the quantification
probability involves the interaction among all uncertain points.
However, the expected nearest neighbor is not a good indicator under
large uncertainty (see \cite{ytxpz-snnsu-10} for details).

Instead of returning only the nearest neighbor, one can ask to return
the $k$ nearest neighbors in a ranked order (the $k$\NN problem).  If
we use expected distance, the ranking of points is straightforward,
namely, rank them in a non-decreasing order of the expected distance
from the query point.  However, when quantification probabilities are
considered, many different criteria for ranking the results are
possible, leading to different problem variants~\cite{jcly-srqpd-11}.

Various combinations of these extensions have been studied in the
literature; see, e.g., \cite{bsi-estkp-08,%
   cccx-eptkn-2009,kcs-cppop-14,%
   ls-aiap-07,ytxpz-snnsu-10}.

\subsection{Our results}
The main results of this paper are the following:
\begin{compactenum}[\quad(i)]
    \item[(i)] A $\Theta(n^3)$ bound on the combinatorial complexity
    $\NZVD(\P)$, an improved quadratic bound on the complexity of
    $\NZVD(\P)$ for a special case, and an efficient randomized
    algorithm for computing $\NZVD(\P)$;
    \item[(ii)] Near-linear-size data structures for answering $\NZNN$
    queries in polylogarithmic or sublinear time;
    \item[(iii)] Efficient data structures for computing the
    quantification probabilities of a query point approximately.
\end{compactenum}
We now describe these results in more detail:

\paragraph{Nonzero Voronoi diagrams.} %
We first study (in \secref{NZVD}) the complexity of
$\NZVD(\P)$. Suppose the uncertainty region of each $P_i \in \P$ is a
disk and $\dist(\cdot, \cdot)$ is the $L_2$ metric. We show that
$\NZVD(\P)$ has $O(n^3)$ complexity, and that this bound is tight in
the worst case even if all uncertainty-region disks have the same
radius.  This significantly improves the bound in \cite{se-gvdus-08}
and closes the problem. We also show that the $O(n^3)$ bound holds for
a much larger class of uncertainty regions, namely, even if each
uncertainty region is a semialgebraic set of constant description
complexity; see \secref{NZVD} for the definition of a semialgebraic
set.

If the disks are pairwise disjoint and the ratio of their radii is at
most $\spread$, then the complexity of $\NZVD(\P)$ is
$O(\spread n^2)$, and we prove a lower bound of $\Omega(n^2)$. Again,
this bound holds for a larger class of uncertainty regions.

We show that if each point in $\P$ has a discrete distribution of
description complexity at most $k$, then $\NZVD(\P)$ has $O(k n^3)$
complexity.

We present a randomized, output-sensitive algorithm for computing
$\NZVD(\P)$ in $O(n^2\log n + \mu)$ expected time, where $\mu$ is the
complexity of $\NZVD(\P)$. $\P$ can be preprocessed into a
point-location structure of size $O(\mu)$ that can answer an $\NZNN$
query in $O(\log n + t)$ time, where $t$ is the output
size~\cite{bcko-cgaa-08}.

\paragraph{Answering $\NZNN$ queries.} %
Since the complexity of $\NZVD(\P)$ can be cubic in the worst case, in
\secref{indexing_schemes} we present near-linear size data structures
for answering $\NZNN$ queries efficiently.  In particular, if the
uncertainty region of each point is a disk then an $\NZNN$ query can
be answered in $O(\log n + t)$ time using $O(n\polylog(n))$ space,
where $t$ is the output size.  If each point of $\P$ has a discrete
distribution of size at most $k$, then an $\NZNN$ query can be
answered in $O(N^{1/2}\log^3 N + t)$ time using $O(N\log^2 N)$ space,
where $t$ is the output size and $N=nk$.  These results rely on
geometric data structures for answering simplex range queries and
their variants; see~\cite{a-rs-16} for a recent survey.

\paragraph{Computing quantification probabilities.} %

Next, in \secref{sec:quanprob}, we focus our attention on computing
quantification probabilities for a query point $q$, i.e., reporting
the values of $\pi_i(q)$ for all $P_i$'s for which $\pi_i(q) > 0$.  We
begin in \secref{exact}, by describing a data structure that can
compute quantification probabilities exactly if each $P_i$ has a
discrete distribution of size at most $k$. Its size is $O(N^4)$ and it
can return all $t$ positive quantification probabilities for a query
point in time $O(\log N+t)$, where $N=nk$ as above. Since computing
quantification probabilities is expensive even for points with
discrete distributions, we mostly focus on computing them
approximately.

We present two data structures for approximating the quantification
probabilities efficiently. The first (see \secref{Monte:Carlo}) is a
Monte-Carlo algorithm for estimating $\pi_i(q)$ for any $P_i$ and $q$
within additive error $\eps$ with probability at least $1 - \delta$,
for parameters $\eps, \delta \in (0, 1)$.  We argue that if each
uncertain point has a discrete distribution of size at most $k$, then
we can estimate $\pi_i(q)$ within additive error $\eps$ with
probability at least $1 - \delta$ by using
$s_{\eps, \delta}= O((1/\eps^2) \log (N/\delta))$ random
instantiations of $\P$. (Note that there are at most $1/\eps$ $P_i$'s
for which $\pi_i(q) > \eps$.)  Consequently, we can preprocess $\P$
into $\anindex$ of size $O((n/\eps^2) \log (N/\delta))$ so that for
any query point $q \in \Re^2$, $\pi_i(q)$ for all $P_i$'s can be
estimated within additive error $\eps$ in
$O((1/\eps^2)\log (N/\delta)\log n)$ time, with probability at least
$1-\delta$.  The algorithm explicitly computes the estimates of
$\pi_i(q)$'s for at most $s_{\eps, \delta}$ points and sets the
estimate to~$0$ for the rest of the points.  We also show that this
approach works even if the distribution of each $P_i$ is continuous,
by approximating a continuous distribution with a discrete one. A key
observation is that it suffices to sample a polynomial number of
points from the distribution of each $P_i$ to ensure that the error in
the quantification probability is at most $\eps$.

Next, in \secref{SpiralSearch} we describe a deterministic algorithm
for computing $\pi_i(q)$ approximately if each point has a discrete
distribution of size at most $k$.  We show that $\P$ can be
preprocessed into $\anindex$ of $O(N)$ size so that for any
$q \in \reals^2$ and for any $\eps \in (0, 1)$, $\pi_i(q)$, for all
$i \in \{1, \ldots, n\}$, can be computed with additive error at most
$\eps$ in $O(\rho k \log (\rho/\eps) + \log N)$ time, where $\rho$ is
the ratio of the largest to the smallest location probabilities over
all possible locations of points in $\P$. We show that there are at
most $m(\rho,\eps)= \rho k \ln (\rho/\eps)+k-1$ points of $\P$ for
which $\pi_i(q) > \eps$. The algorithm explicitly estimates $\pi_i(q)$
for at most $m(\rho,\eps)$ points and sets the estimate to $0$ for the
rest of the points.  

\section{Nonzero probabilistic Voronoi diagram}%
\seclab{NZVD} %

Let $\P$ be a set of $n$ uncertain points as described earlier. We
analyze the combinatorial structure of $\NZVD(\P)$ and describe
algorithms for constructing it. We first consider the case when the
distribution of each point is continuous and then consider the
discrete case.

\subsection{Continuous case} %
\seclab{VD:continuous}

For simplicity, we first assume that the uncertainty region of each
$P_i$ is a circular disk $D_i$ of radius $r_i$ centered at $c_i$.

We first observe that the structure of $\NZVD(\P)$ does not depend on
the actual \pdf of $P_i$'s.  What really matters is the uncertainty
region $D_i$. More precisely, for each $1 \leq i \leq n$ and for
$\query \in \Re^2$, let
\begin{align*}
  &\Delta_i(\query) = \max_{p \in D_i} \dist(\query, p)
    = \dist(\query, c_i) + r_i,\\
    %
  \quad
  &\delta_i(\query) = \min_{p \in D_i} \dist(\query, p) =
    \max\{\dist(\query, c_i) - r_i, 0\}
\end{align*}
be the maximum and minimum possible distance, respectively, from
$\query$ to $P_i$.

The following lemma, whose proof is straightforward, characterizes the
structure of $\NZVD(\P)$.
\begin{lemma}%
    \lemlab{NZobservation}
    For a point $q \in \Re^2$, a point $P_i \in \P$ belongs to
    $\NZNN(q, \P)$ if and only if
    \begin{align*}
      \delta_i(q) < \Delta_j(q) \text{ for all } 1 \leq j\neq i \leq
      n.
    \end{align*}
\end{lemma}

Let $\Delta\colon \Re^2 \rightarrow \Re$ denote the \emphi{lower
   envelope}%
\footnote{The \emphi{lower envelope}, $L_F$, of a set $F$ of functions
   is their pointwise minimum, i.e., $L_F(x) = \min_{f\in F} f(x)$.
   The \emphi{upper envelope}, $U_F$, of $F$ is the pointwise maximum,
   i.e., $U_F(x) = \max_{f\in F} f(x)$.} %
of $\Delta_1, \ldots$, $\Delta_n$; that is, for any
$\query \in \Re^2$,
\begin{align*}
  \Delta(\query) = \min_{1\leq i \leq n} \Delta_i(\query).
\end{align*}

The projection of the graph of $\Delta(x)$ onto the $xy$-plane is the
additive-weighted Voronoi diagram of the points $c_1, \ldots, c_n$,
where the weight of $c_i$ is $r_i$, and the weighted distance from
$\query$ to $c_i$ is $\dist(\query, c_i) + r_i$, for $i=1,\ldots, n$.
Let $\mathbb{M}$ denote this planar subdivision. It has linear
complexity and each of its edges is a hyperbolic arc;
see~\cite{ab-gdt-86}. \lemref{NZobservation} implies that, for any
$\query \in \Re^2$,
\begin{align}%
  \eqlab{NZNN:obs}%
  \NZNN(\query, \P) = \brc{P_i \mid \delta_i(\query) <
  \Delta(\query)}.
\end{align}

\begin{figure}
    \centering
    \begin{tabular}{ccc}
      \IncludeGraphics[page=1,width=0.25\linewidth]%
      {\si{figs/nznn_observation}}
      & \hspace{15mm} &%
                        \IncludeGraphics[page=2,width=0.25\linewidth]%
                        {\si{figs/nznn_observation}}
    \end{tabular}%
    \caption{$\P = \{P_1, \ldots, P_5\}$, $\Delta(x) = \Delta_1(x)$,
       $\NZNN(x, \P) = \{P_1, P_2, P_3\}$,
       $\Delta(x') = \Delta_1(x')$, $\NZNN(x', \P) = \{P_1, P_2\}$,
       and $x'$ lies on an edge of $\NZVD(\P)$.}
    \figlab{NZNN:observation}
\end{figure}

See \figref{NZNN:observation}. It also implies that, as we move $x$
continuously in $\Re^2$, $\NZNN(x, \P)$ remains the same until
$\delta_i(x)$, for some $1\leq i \leq n$, becomes equal to $\Delta(x)$
(e.g., $x'$ in \figref{NZNN:observation}).  This observation was made
in earlier papers as well; see,
e.g.~\cite{ccmc-pvecn-08,ckp-qidmo-04}.  Using this observation we can
now characterize $\NZVD(\P)$.

For $i=1,\ldots, n$, let
$\gamma_i = \{ x \in \Re^2 \mid \delta_i(x) = \Delta(x)\}$ be the zero
set of the function $\Delta(x) - \delta_i(x)$. Set
$\Gamma = \{\gamma_1, \ldots, \gamma_n\}$.

The curve $\gamma_i$ partitions the plane into two open regions:
$\Delta(x) < \delta_i(x)$ and $\Delta(x) > \delta_i(x)$. By
\Eqref{NZNN:obs}, $P_i \in \NZNN(x, \P)$ for all points $x$ inside the
latter region and for none of the points $x$ inside the former
region. It is well known that, for any fixed $j \neq i$,
$\gamma_{ij} = \{ x \in \Re^2 \mid \delta_i(x) = \Delta_j(x)\}$ is a
hyperbolic curve~\cite{ab-gdt-86}. The curve $\gamma_i$ is composed of
pieces of $\gamma_{ij}$, for $j \neq i$. We refer to the endpoints of
these pieces as \emphi{breakpoints} of $\gamma_i$. They are the
intersection points of $\gamma_i$ with an edge of $\mathbb{M}$ and
correspond to points $q$ such that the disk of radius $\Delta(q)$
centered at $q$ touches (at least) two disks of $\D$ from inside,
touches $D_i$ from outside, and does not contain any disk of $\D$ in
its interior. See \figref{Two:Inside:One:Outside}.  Formally, we say
that a disk $D_1$ touches a disk $D_2$ from the \emphi{outside}
(resp.\ \emphi{inside}) if
$\partial D_1 \cap \partial D_2 \neq \emptyset$ and
$\Int D_1 \cap \Int D_2 = \emptyset$ (resp.\
$\Int D_2 \subseteq \Int D_1$).

\begin{figure}[t]
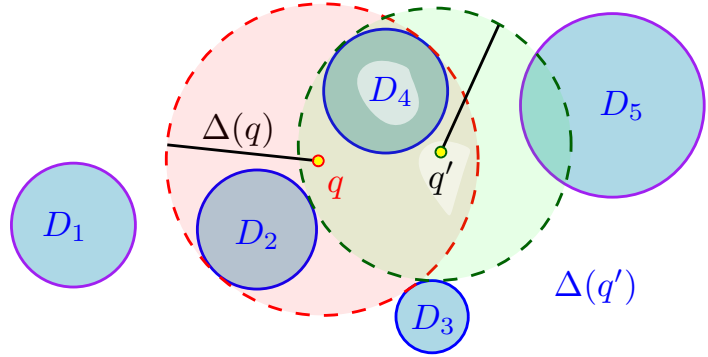

    \centering
    {\IncludeGraphics[width=0.5\linewidth]{figs/\si{two_inside_one_outside}}}
    \caption{The point $q$ is a break point of $\gamma_3$ and $q'$ is
       an intersection point of $\gamma_2$ and $\gamma_3$.}
    \figlab{Two:Inside:One:Outside}
\end{figure}

\begin{lemma}%
    \lemlab{breakpoints}%
    The curve $\gamma_i$, $1 \leq i \leq n$, has at most $2n$
    breakpoints, and it can be computed in $O(n\log n)$ time.
\end{lemma}
\begin{proof}
    Let $\Gamma_i = \{ \gamma_{ij} \mid j \neq i, 1 \leq j \leq n \}$.
    It can be verified that a ray emanating from $c_i$ intersects the
    hyperbolic curve $\gamma_{ij}$, for any $j\neq i$, in at most one
    point, so $\gamma_{ij}$ can be viewed as the graph of a function
    in polar coordinates with $c_i$ as the origin. That is, let
    $\gamma_{ij}\colon [0, 2\pi) \rightarrow \Re_{\geq 0}$, where
    $\gamma_{ij}(\theta)$ is the distance from~$c_i$ to~$\gamma_{ij}$
    in direction $\theta$. Then, $\gamma_i$ is the lower envelope of
    $\Gamma_i$. Since each pair of curves in $\Gamma_i$ intersects at
    most twice, a well-known result on lower envelopes implies that
    $\gamma_i$ has at most $2n$ breakpoints, and that it can be
    computed in $O(n\log n)$ time~\cite{sa-dsstg-95}. See
    \figref{gammai} for an example.
\end{proof}

Let $\EuScript{A}(\Gamma)$ denote the planar subdivision induced by
$\Gamma$: its vertices are the breakpoints of $\gamma_i$'s and the
intersection points of two curves in $\Gamma$, its edges are the
portions of $\gamma_i$'s between two consecutive vertices, and its
cells are the maximal connected regions of the plane that do not
intersect any curve of $\Gamma$. We refer to vertices, edges, and
cells of $\EuScript{A}(\Gamma)$ as its 0-, 1-, and 2-dimensional
\emphi{faces}.

\begin{figure*}
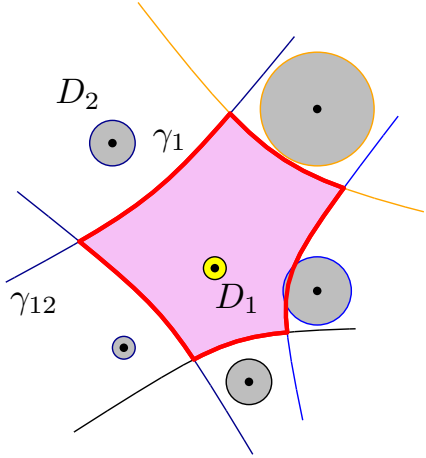

    \centering
    \IncludeGraphics[width=0.30\textwidth]%
    {\si{figs/vor_cell}}
    \caption{An example of $\gamma_1$.}
    \figlab{gammai}
\end{figure*}

For a face $\phi$ (of any dimension), and for any two points
$x , y \in \phi$, the sets $\{P_i \mid \delta_i(x) < \Delta(x)\}$ and
$\{P_j \mid \delta_j(y) < \Delta(y)\}$ are the same; we denote this
set by $\P_{\phi}$. Furthermore, if $x, y$ lie in two neighboring
faces $\phi$ and $\phi'$, respectively, then
$\P_{\phi} \neq \P_{\phi'}$. The following lemma is an immediate
consequence of \Eqref{NZNN:obs}.

\begin{lemma}
    For all points $x$ lying in a face $\phi$
    of~$\EuScript{A}(\Gamma)$, $\NZNN(x, \P) = \P_{\phi}$.
\end{lemma}

For a subset $\T \subseteq \P$, let $\NZVorCell(\T)$ be as defined in
\Eqref{NZVorCell}.  An immediate corollary of the above lemma is:
\noindent%
\begin{corollary}
    \corlab{NZVD}
    \begin{inparaenum}[(i)]
        \item For any $\T \subseteq \P$,
        $\NZVorCell(\T) \neq \emptyset$ if and only if there is a face
        $\phi$ of $\EuScript{A}(\Gamma)$ with $\T = \P_{\phi}$.

        \item The planar subdivision $\EuScript{A}(\Gamma)$ coincides
        with $\NZVD(\P)$.
    \end{inparaenum}
\end{corollary}

We now bound the complexity of $\EuScript{A}(\Gamma)$ and thus of
$\NZVD(\P)$.

\begin{theorem}%
    \thmlab{continuous1}%
    Let $\P = \{P_1, \ldots, P_n \}$ be a set of $n$ uncertain points
    in $\Re^2$ whose uncertainty regions are disks. Then $\NZVD(\P)$
    has $O(n^3)$ complexity. Moreover, it can be computed in
    $O(n^2\log n + \mu)$ expected time, where $\mu$ is the complexity
    of $\NZVD(\P)$.
\end{theorem}

\begin{proof}
    Using a standard perturbation argument (see,
    e.g.,~\cite{sa-dsstg-95}), it suffices to bound the complexity of
    $\NZVD(\P)$ when the disks corresponding to the uncertainty
    regions of the points of $\P$ are in general position, so we can
    assume that the degree of every vertex in $\NZVD(\P)$ is
    constant. Since $\NZVD(\P)$ is a planar subdivision and the degree
    of every vertex is constant, the number of edges and cells in
    $\NZVD(\P)$ is proportional to the number of its vertices. Hence,
    it suffices to bound the number of vertices. Let
    $\Gamma = \{\gamma_1, \ldots, \gamma_n \}$ be the set of curves as
    defined above. By \lemref{breakpoints}, each $\gamma_i$ has $O(n)$
    breakpoints, so there are a total of $O(n^2)$ breakpoints. We
    claim that each pair of curves $\gamma_i$ and $\gamma_j$ intersect
    $O(n)$ times --- each such intersection point corresponds to a
    point $v \in \Re^2$ such that the disk of radius $\Delta(v)$
    centered at $v$ touches $D_i$ and $D_j$ from the outside and
    another disk $D_k$ of $\D$, the one realizing the value of
    $\Delta(v)$, from the inside (e.g., $q'$ in
    \figref{Two:Inside:One:Outside}). For a fixed $k$, we show that
    there are at most two points $v$ such that
    $\delta_i(v) = \delta_j(v) = \Delta_k(v)$. Note that
    $\delta_i(v) = \Delta_k(v)$ represents either an empty set or one
    hyperbolic branch, and the same holds for
    $\delta_j(v) = \Delta_k(v)$. Two such hyperbolic branches
    intersect at most twice, implying that
    $\delta_i(v) = \delta_j(v) = \Delta_k(v)$ contributes at most two
    vertices. Hence, the number of vertices in $\NZVD(\P)$ is
    $O(n^3)$, as claimed.

    By \lemref{breakpoints}, $\Gamma$ can be computed in
    $O(n^2\log n)$ time.  The planar subdivision $\A(\Gamma)$ can be
    computed in $O(n\log n + \mu)$ expected time using randomized
    incremental method \cite{as-aa-00}, where $\mu$ is the complexity
    of $\NZVD(\P)$. Hence $\NZVD(\P)$ can be computed in
    $O(n^2\log n + \mu)$ expected time.
\end{proof}

The above argument is quite general and extends to a large class of
uncertainty regions. In particular, a two-dimensional
\emphi{semialgebraic set} is a subset of $\Re^2$ obtained from a
finite number of sets of the form $\{x\in\reals^2\mid g(x)\ge 0\}$,
where $g$ is a bivariate polynomial with real coefficients, by Boolean
operations (union, intersection, and complement).  A semialgebraic set
has \emph{constant description complexity} if the number of
polynomials defining the set as well as the maximum degree of these
polynomials is a constant. For example, a polygon with constant number
of edges and a region defined by a constant number of quadratic arcs
are semialgebraic sets of constant description complexity.

Suppose the uncertainty region of each point in $\P$ is a
semialgebraic set of constant description complexity, denoted by
$\sigma_i$.

The analysis for the case of disks shows that a vertex of $\NZVD(\P)$
is the center of a disk that touches uncertainty regions of three
different points. Fix a triple of uncertainty regions
$\sigma_i, \sigma_j, \sigma_k$. Since they are semialgebraic sets of
constant complexity, there are only $O(1)$ disks that are tangent to
$\sigma_1, \sigma_2$, and $\sigma_3$ simultaneously.  Therefore,
$\NZVD(\P)$ has $O(n^3)$ vertices, which in view of the above
discussion implies that $\NZVD(\P)$ has $O(n^3)$ combinatorial
complexity. Assuming an extension of the real RAM model of computation
in which the roots of polynomials of constant degree can be computed
exactly in $O(1)$ time, the randomized algorithm described above can
be extended to this case as well. Omitting further details, we
conclude the following:

\begin{theorem}%
    \thmlab{continuous-semi}%
    Let $\P = \{P_1, \ldots, P_n \}$ be a set of $n$ uncertain points
    in $\Re^2$ whose uncertainty regions are semialgebraic sets of
    constant description complexity. Then $\NZVD(\P)$ has $O(n^3)$
    complexity. Moreover, it can be computed in $O(n^2\log n + \mu)$
    expected time, where $\mu$ is the complexity of $\NZVD(\P)$.
\end{theorem}

\begin{figure}[t]
    \centering
    \IncludeGraphics[page=1]%
    {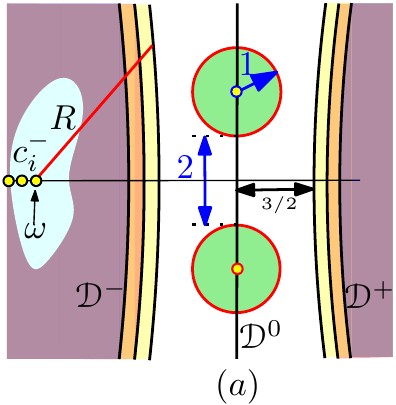}%
    \hspace{25mm}%
    \IncludeGraphics[page=2]%
    {figs/lower_bound_n_3}
    \vspace{-3mm}
    \caption{(a) $\Omega(n^3)$ lower bound construction with $m = 3$;
       only some disks are drawn. (b) Illustration of the proof. }
    \figlab{n:3:lower:bound:construction}
\end{figure}
\begin{theorem}
    There exists a set $\P$ of $n$ uncertain points whose uncertainty
    regions are disks such that $\NZVD(\P)$ has $\Omega(n^3)$
    vertices.
\end{theorem}
\begin{proof}
    Assume that $n = 4m$ for some $m \in \mathbb{N}^+$. We choose two
    parameters $R = 8n^2$ and $\omega = 1/n^2$. We construct three
    families of disks: $\D^- = \{D^-_1,\ldots, D^-_m\}$,
    $\D^+ = \{D^+_1,\ldots, D^+_m\}$, and
    $\D^0 = \{D^0_1,\ldots, D^0_{2m}\}$.  The radius of all disks in
    $\D^- \cup \D^+$ is $R$ and their centers lie on the $x$-axis; the
    radius of all disks in $\D^0$ is 1 and their centers lie on the
    $y$-axis. More precisely, for $1\leq i, j\leq m$, the center of
    $D^-_i$ is $c^-_i = (-R - 3/2 - (i-1)\omega, 0)$ and the center of
    $D^+_j$ is $c^+_j = (R + 3/2 + (j-1)\omega, 0)$, and for
    $1\leq k \leq 2m$, the center of $D^0_k$ is $(0, 4(k - m) - 2)$.
    See \figref{n:3:lower:bound:construction}(a).

    We claim that for every triple $i, j, k$ with $1 \leq i, j \leq m$
    and $1\leq k \leq 2m$, there are two disks each of which touches
    $D^-_i$ and $D^+_j$ from the outside and $D^0_k$ from the inside
    and does not contain any disk of $\D^-\cup \D^+ \cup \D^0$ in its
    interior. See \figref{n:3:lower:bound:construction}(b).

    Fix such a triple $i, j, k$. Since the radii of $D^-_i$ and
    $D^+_j$ are the same, the locus $b_{ij}$ of the centers of disks
    that simultaneously touch $D^-_i$ and $D^+_j$ from the outside is
    the bisector of their centers, i.e., $b_{ij}$ is the vertical line
    $x = (x(c^-_i) + x(c^+_j))/2 = (j-i)\omega/2$. Let $\sigma_{ij}$
    denote the intersection point of $b_{ij}$ and the $x$-axis;
    $\sigma_{ij} = (\frac{1}{2}(j-i)\omega, 0)$. A point on $b_{ij}$
    can be represented by its $y$-coordinate; we will not distinguish
    between the two. For $y$-value $a$, let $\xi_a$ be the disk
    centered at $a$ and simultaneously touching $D^-_i$ and $D^+_j$
    from the outside. The radius of $\xi_a$ is
    \begin{align*}
      \| a - c^-_i\| - R%
      =%
      \sqrt{a^2 + \|c^-_i - \sigma_{ij}\|^2} - R 
      = %
      \sqrt{a^2 + \pth[]{R + 3/2 + \pth[]{\frac{i+j}{2} -
      1}\omega}^2} - R.
    \end{align*}

    The radius of $\xi_a$ is thus at least $3/2$, and for
    $a \in [-4m, 4m]$, it is at most 2 (using the fact that
    $R\geq 8n^2$ and $\omega = 1/n^2$). Hence, for $a \in [-4m, 4m]$,
    $\xi_a$ contains at most one disk of $\D^0$ in its interior, and
    obviously $\xi_a$ does not contain any disk of $\D^- \cup \D^+$ in
    its interior.

    Let $a_k = 4(k - m) - 2$. Then, the disk $\xi_{a_k}$ contains
    $D^0_k$ in its interior because the distance between the centers
    of $D^0_k$ and $\xi_{a_k}$ is at most $m\omega \leq 1/(4n)$, the
    radius of $D^0_k$ is 1, and the radius of $\xi_{a_k}$ is at least
    $3/2$. On the other hand, the disk $\xi_a$ for $a = a_k \pm 2$
    does not contain $D^0_k$ in its interior because the radius of
    $\xi_a$ is at most 2 and the distance between the center of
    $D^0_k$ and $\xi_a$ is at least 2. Therefore, by a continuity
    argument, there is a value $a^+ \in [a_k, a_k + 2]$ at which
    $\xi_{a^+}$ touches $D^0_k$ from the inside. Similarly, there is a
    value $a^- \in [a_k-2, a_k]$ at which $\xi_{a^-}$ touches $D^0_k$
    from the inside.

    This proves the claim that there are two disks touching $D^-_i$
    and $D^+_j$ from the outside and $D^0_k$ from the inside and not
    containing any disk of $\D^-\cup \D^+ \cup \D^0$ in its
    interior. In other words, each triple $i,j,k$ contributes two
    vertices to $\NZVD(\P)$. Hence, $\NZVD(\P)$ has $\Omega(n^3)$
    vertices.
\end{proof}

Next, we show that the maximum complexity of $\NZVD(\P)$ is
$\Omega(n^3)$ even if the uncertainty regions of points in~$\D$ are
disks of the same radius.

\begin{theorem}
    There exists a set $\P$ of $n$ uncertain points, whose uncertainty
    regions are disks of the same radius, for which $\NZVD(\P)$ has
    $\Omega(n^3)$ vertices.
\end{theorem}

\begin{figure*}
    \centering
    \begin{tabular}{ccc}

        \IncludeGraphics[width=0.45\textwidth]%
        {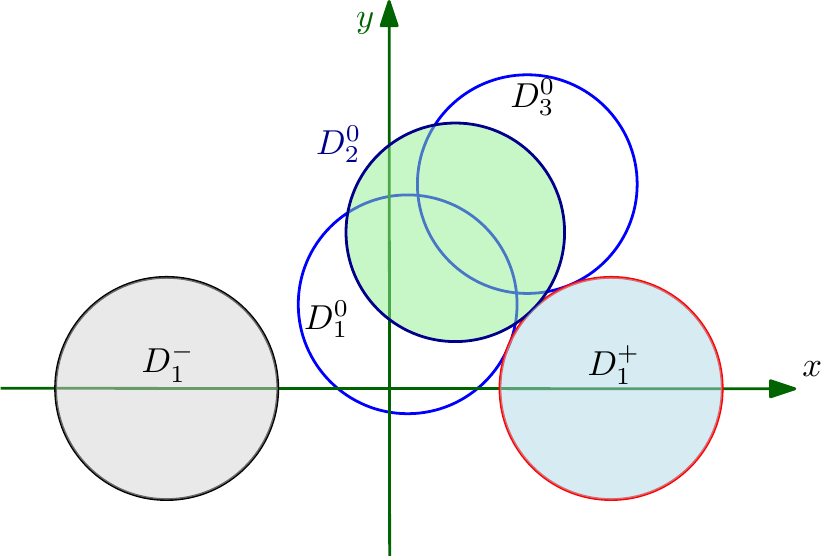}
      &\hspace*{-5mm}&
                       \IncludeGraphics[width=0.45\textwidth]%
                       {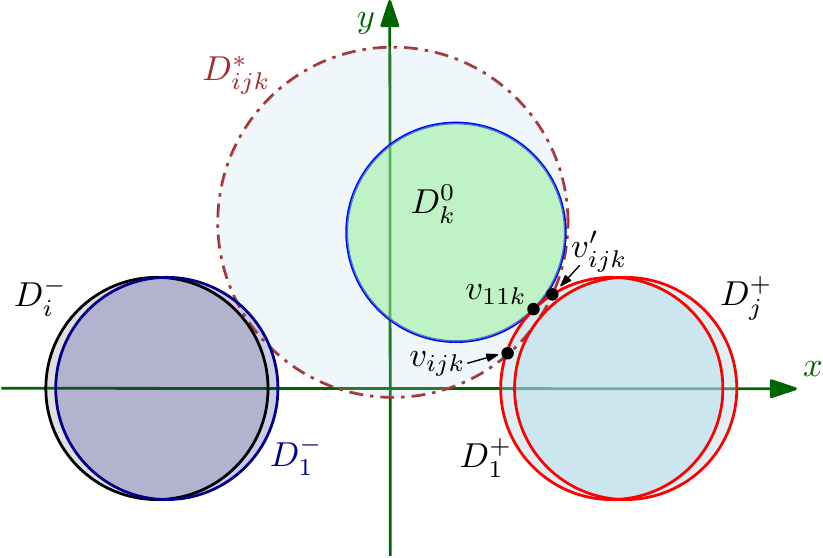} \\
      && \\
      \small (a) && \small (b)
    \end{tabular}
    \caption{(a) $\Omega(n^3)$ lower bound construction using disks of
       same radius with $m = 3$; only some disks are drawn. (b)
       Illustration of the proof.}
    \figlab{n:3:lower:bound:construction:same:radius}
\end{figure*}

\begin{proof}
    Assume that $n = 3m$ for some $m \in \mathbb{N}^+$. We choose two
    parameters $\theta = \frac{\pi}{2}\cdot \frac{1}{(m+1)}$, and a
    sufficiently small positive number $\omega$. We construct three
    families of disks: $\D^- = \{D^-_1,\ldots, D^-_m\}$,
    $\D^+ = \{D^+_1,\ldots, D^+_m\}$, and
    $\D^0 = \{D^0_1,\ldots, D^0_{m}\}$. Without loss of generality, we
    set the radius of all disks to~$1$. The centers of disks in
    $\D^- \cup \D^+$ lie on the $x$-axis, and the centers of disks in
    $\D^0$ lie in the first quadrant.  More precisely, for
    $1\leq i, j\leq m$, the center of $D^-_i$ is
    $c^-_i = (-2 - (i-1)\omega, 0)$ and the center of $D^+_j$ is
    $c^+_j = (2 + (j-1)\omega, 0)$, and for $1\leq k \leq m$, the
    center of $D^0_k$ is $(2 - 2\cos(k\theta), 2\sin(k\theta))$. See
    \figref{n:3:lower:bound:construction:same:radius}(a).

    We claim that for every triple $i, j, k$ with
    $1 \leq i, j, k\leq m$, there is a disk touching $D^-_i$ and
    $D^+_j$ from the outside and $D^0_k$ from the inside and not
    containing any disk of $\D^-\cup \D^+ \cup \D^0$ in its interior.

    First of all, this is true for $i = j = 1$ and $1\leq k \leq m$.
    Note that $D^-_1$ is centered at $(-2, 0)$, $D^+_1$ is centered at
    $(2, 0)$, and $D^0_k$ touches $D^+_1$ from the outside. Since the
    radius of $D^-_1$ and $D^+_1$ is the same, the locus $b_{11}$ of
    the centers of disks that simultaneously touch $D^-_1$ and $D^+_1$
    from the outside is the bisector of their centers, i.e., $b_{11}$
    is $y$-axis. Fix a value $k$. It is easy to see that the disk
    $D^*_{11k}$ centered at $(0, 2\tan(k\theta))$ with the radius
    $\frac{2}{\cos(k\theta)} - 1$ touches $D^-_1$ and $D^+_1$ from the
    outside and $D^0_k$ from the inside. Furthermore, we show that
    $D^*_{11k}$ does not contain disks in $\D^0$ in its interior
    (obvious for $\D^-\cup \D^+$). Since every disk in $\D^0$ touches
    $D^+_1$ from the outside, $D^*_{11k}$ containing a disk of $\D^0$
    in its interior would imply that $D^*_{11k}$ intersects the
    interior of $D^+_1$, a contradiction.

    Next, we show that it holds for $1 < i, j \leq m$ and
    $1 \leq k \leq m$. The key idea is that $D^-_i$ (resp. $D^+_j$)
    got placed by translating (``perturbing'') $D^-_1$ (resp. $D^+_1$)
    so little that the disk $D^*_{ijk}$ touching $D^-_i$ and $D^+_j$
    from the outside and $D^0_k$ from the inside does not contain any
    disk of $\D^-\cup \D^+ \cup \D^0$ in its interior as for the case
    when $i = j = 1$. We argue this using some elementary
    geometry. See
    \figref{n:3:lower:bound:construction:same:radius}(b). Let
    $v_{ijk}$ and $v'_{ijk}$ be the two intersection points of
    $\partial D^*_{ijk}$ and $\partial D^+_1$, for
    $1 \leq i, j, k \leq m$. Such two intersection points coincide
    with each other when $i = j = 1$, and furthermore, they always
    exist due to our construction. Note that $v_{11k}$ is also the
    intersection point of $\partial D^0_k$ and $\partial D^+_1$. It is
    trivial to see that as the parameter $\omega$ gets smaller,
    $v_{ijk}$ and $v'_{ijk}$ lie closer to $v_{11k}$ along
    $\partial D^+_1$. They all coincide with $v_{11k}$ when $\omega$
    becomes 0. Since $\omega$ is a sufficiently small positive number,
    we are assured that $v_{ijk}$ lies between $v_{11(k-1)}$ and
    $v_{11k}$ along $\partial D^+_1$, i.e., $D^*_{ijk}$ does not
    contain $D^0_{k-1}$, not to mention $D^0_1, \ldots, D^0_{k-2}$.
    Similarly, $D^*_{ijk}$ does not contain
    $D^0_{k+1}, \ldots, D^0_m$. Moreover, $D^*_{ijk}$ does not contain
    any disk of $\D^- \cup \D^+$. Hence, there is a disk touching
    $D^-_i$ and $D^+_j$ from the outside and $D^0_k$ from the inside
    and not containing any disk of $\D^-\cup \D^+ \cup \D^0$ in its
    interior, for $1 \leq i, j, k \leq m$.

    This proves our claim, and finishes our $\Omega(n^3)$ lower bound
    construction when the disks have the same radius.
\end{proof}

We prove a refined bound on the complexity of $\NZVD(\P)$ if the
uncertainty regions in $\D$ are pairwise-disjoint disks and the ratio
of the radii of the largest to the smallest disk is bounded by
$\spread$.

\begin{lemma}%
    \lemlab{rho}%
    If $\P = \brc{P_1, \ldots, P_n }$ is a set of $n$ uncertain points
    in $\Re^2$ whose uncertainty regions are pairwise-disjoint disks
    with radii in the range $[1,\spread]$, a pair of curves in
    $\Gamma$ intersects in $O(\spread)$ points.
\end{lemma}

\begin{proof}
    Fix a pair of curves $\gamma_1$ and $\gamma_2$ in $\Gamma$. Let
    $D_1$ and $D_2$ be the disks corresponding to $\gamma_1$ and
    $\gamma_2$, and let $c_1$ and $c_2$ be their centers,
    respectively. By applying rotation and translation to the plane,
    we can assume $D_1$ and $D_2$ are centered on the $x$-axis, with
    $D_1$ to the left of $D_2$.

    For a parameter $t$, $1 \leq t \leq \spread$, let $\DiskSet$
    denote the set of all the disks associated with $\P$, excluding
    $D_1$ and $D_2$, with radii between $t$ and $2t$.  An intersection
    point $\query \in \gamma_1 \cap \gamma_2$ corresponds to a
    \emphi{witness} disk $\diskW$ centered at~$\query$ that touches
    both $D_1$ and $D_2$ from the outside, touches exactly one other
    disk $\diskA \in \DiskSet$ from the inside, and properly contains
    no disks of~$\DiskSet$. The family of disks that touch both
    $\disk_1$ and $\disk_2$ from the outside is a \emphi{pencil},
    which sweeps over a portion of the plane as the tangency points
    with $D_1$ and $D_2$ move continuously and monotonically in the
    $y$-direction.
    A disk of $\DiskSet$ can contribute at most two intersection
    points to $\gamma_1\cap \gamma_2$, as its boundary gets swept over
    at most twice by the circles of the pencil.

    We break $\partial{\diskW}$ into two curves, \emphi{top} and
    \emphi{bottom}, at $\diskW$'s tangency points with $D_1$ and
    $D_2$. For a disk $\diskA \in \DiskSet$, if its tangency point
    with its witness disk $\diskW$ is on the top portion of $\diskW$,
    then it is a \emphi{top tangency event}, otherwise it is a
    \emphi{bottom tangency event}. See \figref{bounded_ratio}(a). Let
    $\DiskSet_1$ (resp. $\DiskSet_2$) be the set of disks in
    $\DiskSet$ that are closer to $D_1$ (resp. $D_2$).

    Below we show that the number of top tangency events involving
    disks in $\DiskSet_2$ is $O(\spread/t)$.  Other tangency events
    are handled by a symmetric argument.

    We remove from $\DiskSet_2$ all the disks within distance
    $ T = \constA t $ from $D_2$, where $\constA$ is a sufficiently
    large constant.  The ring with outer radius $\radiusX{D_2} + 4T$
    and inner radius $\radiusX{D_2}$ has area
    \begin{align*}
      \alpha = \pi \pth{ (\radiusX{D_2} + 4T)^2 - ( \radiusX{D_2})^2
      } = O( T \radiusX{D_2} + T^2 ) = O(t^2 + \spread t),
    \end{align*}
    as $\radiusX{D_2} \leq \spread$.  Disks removed from $\DiskSet_2$
    have the following properties:


    \begin{figure*}[t]
        \begin{minipage}[b]{1.0\textwidth}
            \centering
            \begin{tabular}{ccc}
              \IncludeGraphics[scale=0.91,page=1]{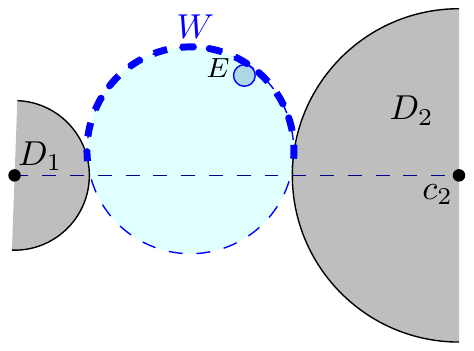}
              & \hspace*{-2mm} & %
                                 \IncludeGraphics[scale=0.91,page=3]{figs/tangency} \\
              (a) & & (b) \\
              \IncludeGraphics[scale=0.91,page=2]{figs/tangency} 
              & \hspace*{-2mm} & %
                                 \IncludeGraphics[scale=0.91,page=4]{figs/tangency} \\
              (c) & & (d) 
            \end{tabular}%
            \caption{An illustration for the proof of \lemref{rho}.}
            \figlab{bounded_ratio}
        \end{minipage}
    \end{figure*}

    \begin{compactenum}[\quad(i)]
        \item they are interior-disjoint,
        \item their radii lie in the interval $[t,2t]$,
        \item they are contained in the aforementioned ring, and
        \item the area of each such disk is at least $\pi t^2$.
    \end{compactenum}
    Hence, the number of removed disks is
    $ O( (t^2 + \spread t) / t^2 ) = O( \spread / t) $.

    Consider the circle $\Circle_2$ of radius $\radiusX{D_2} + T/2$
    centered at~$c_2$. Consider any disk $\diskA \in \DiskSet_2$ and
    its witness disk $\diskW$ touching both $D_1$ and $D_2$ from the
    outside. If $\diskA$ has not been removed from~$\DiskSet_2$, then
    $\radiusX{\diskW}\geq (T+2t)/2$; in particular it is larger than
    $T/2$ and the center of $W$ lies outside $\Circle_2$.  Let $\Wtau$
    be the disk concentric with $\diskW$ with radius
    $\radiusX{\diskW}-\tau$, where $\tau = 4t$.  The interior of
    $\Wtau$ is disjoint from all disks in $\DiskSet_2$, as $\diskA$
    touches $\diskW$ from inside and $\diskW$ does not fully contain
    any other disks from $\DiskSet_2$.  \vspace{-0.15cm} %
    \parpic[r]{\IncludeGraphics{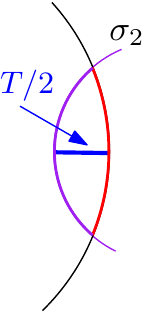}}%
    The witness disk $\diskW$ covers an arc of length at least $T/2$
    on~$\Circle_2$. Indeed, neither of these two disks contains the
    center of the other, and the inner distance between the two
    intersection arcs is $T/2$, see figure on the right. Similarly,
    let $\IntervalX{\diskA}$ be the arc $\Wtau \cap \Circle_2$.  By
    the same argument, we have that $\IntervalX{\diskA}$ is of length
    at least $T/2 - \tau = \Omega(t)$.

    The circumference of $\Circle_2$ is
    $2 \pi(\radiusX{D_2}+T/2) = O( \spread)$, so if the arcs
    $\IntervalX{\diskA}$, for $\diskA \in \DiskSet_2$, are pairwise
    disjoint, we are done, as this implies that there could be at most
    $\spread/(T/2 -\tau) = O( \spread / t)$ such arcs and thus the
    size of the original $\DiskSet_2$, including the disks that were
    deleted from $\DiskSet_2$ is $O( \spread / t)$. See
    \figref{bounded_ratio}(b).

    We now prove the claim that for any two disks
    $\diskA, \diskA' \in \DiskSet_2$ realizing a top tangency event,
    $\IntervalX{\diskA}$ and $\IntervalX{\diskA'}$ are disjoint.

    Let $\diskW$ (resp.\ $\diskW'$) be the witness disk that is
    tangent to $D_1, D_2$ and $\diskA$ (resp.\ $\diskA'$).  Assume
    that the tangency of $\diskW$ with $D_2$ is clockwise to the
    tangency of $\diskW'$ with $D_2$ (i.e., $\diskA$ is ``above''
    $\diskA'$).  If $\Wtau$ and $\Wtau'$ are disjoint then the
    corresponding arcs $\IntervalX{\diskA}$ and $\IntervalX{\diskA'}$
    are obviously disjoint, so assume that $\Wtau$ and $\Wtau'$
    intersect; see \figref{bounded_ratio}(c).
    
    Let $c'$ be the center of $\Wtau'$.  We define three circular arcs
    on $\partial\Wtau'$.  Let $\xi_1=\partial\Wtau' \cap \Wtau$, let
    $\xi_2$ be the portion of $\partial\Wtau'$ lying in the disk
    bounded by $\sigma_2$, and let $\xi_3$ be the portion of
    $\partial \Wtau'$ lying in the wedge formed by the rays $c'c_1$
    and $c'c_2$; see \figref{bounded_ratio}(d).  It can be verified
    that $\xi_1 \subset \xi_3$ and the right endpoint of $\xi_3$ lies
    inside $\sigma_2$ and thus on $\xi_2$.
    
    Next, let $\eta \in \partial\Wtau'$ be the intersection point of
    $\partial\Wtau'$ with the segment connecting $c'$ and the center
    of $\diskA'$; since $\diskA'$ lies in the exterior of $\Wtau'$,
    $\eta$ exists. Since $\diskA'$ realizes a top tangency event,
    $\eta \in \xi_3$. Furthermore, $\diskA'$ lies in the exterior of
    $\Wtau$ and $\diskA' \in \DiskSet_2$, which implies that
    $\eta \not\in \xi_1$ and it lies to the right of $\xi_1$.
    Similarly, $\diskA'$ lies in the exterior of $\sigma_2$ and the
    right endpoint of $\xi_3$ lies on $\xi_2$, therefore $\eta$ lies
    to the left of the arc $\xi_2$. In other words, $\eta$ separates
    $\xi_1$ and $\xi_2$, implying that $\xi_1\cap\xi_2 = \emptyset$,
    which in turn implies that the top endpoint of
    $\IntervalX{\diskA'}$ does not lie inside $\Wtau$. Hence,
    $\IntervalX{\diskA} \cap \IntervalX{\diskA'} = \emptyset$, as
    claimed.


    Finally, We repeat the above counting argument, for
    $t=1,2,4,\ldots, 2^m$, where $m=\ceil{\log_2 \spread}$, concluding
    that the number of intersection points between $\gamma_1$ and
    $\gamma_2$ is bounded by
    $\sum_{i=1}^m O( \spread/2^i) = O(\spread)$. This completes the
    proof of the lemma.
\end{proof}

\begin{theorem}%
    \thmlab{continuous:2}
    Let $\P = \{P_1, \ldots, P_n \}$ be a set of $n$ uncertain points
    in $\Re^2$ such that their uncertainty regions are
    pairwise-disjoint disks and that the ratio of the largest and the
    smallest radii of the disks is at most $\spread$. Then, the
    complexity of $\NZVD(\P)$ is $O(\spread n^2)$, and it can be
    computed in $O(n^2\log n + \mu)$ expected time, where $\mu$ is the
    complexity of $\NZVD(\P)$. Furthermore, there exists such a set
    $\P$ of uncertain points for which $\NZVD(\P)$ has $\Omega(n^2)$
    complexity.
\end{theorem}

\begin{figure}[t]
    \centering
    \IncludeGraphics[width=0.60\textwidth]%
    {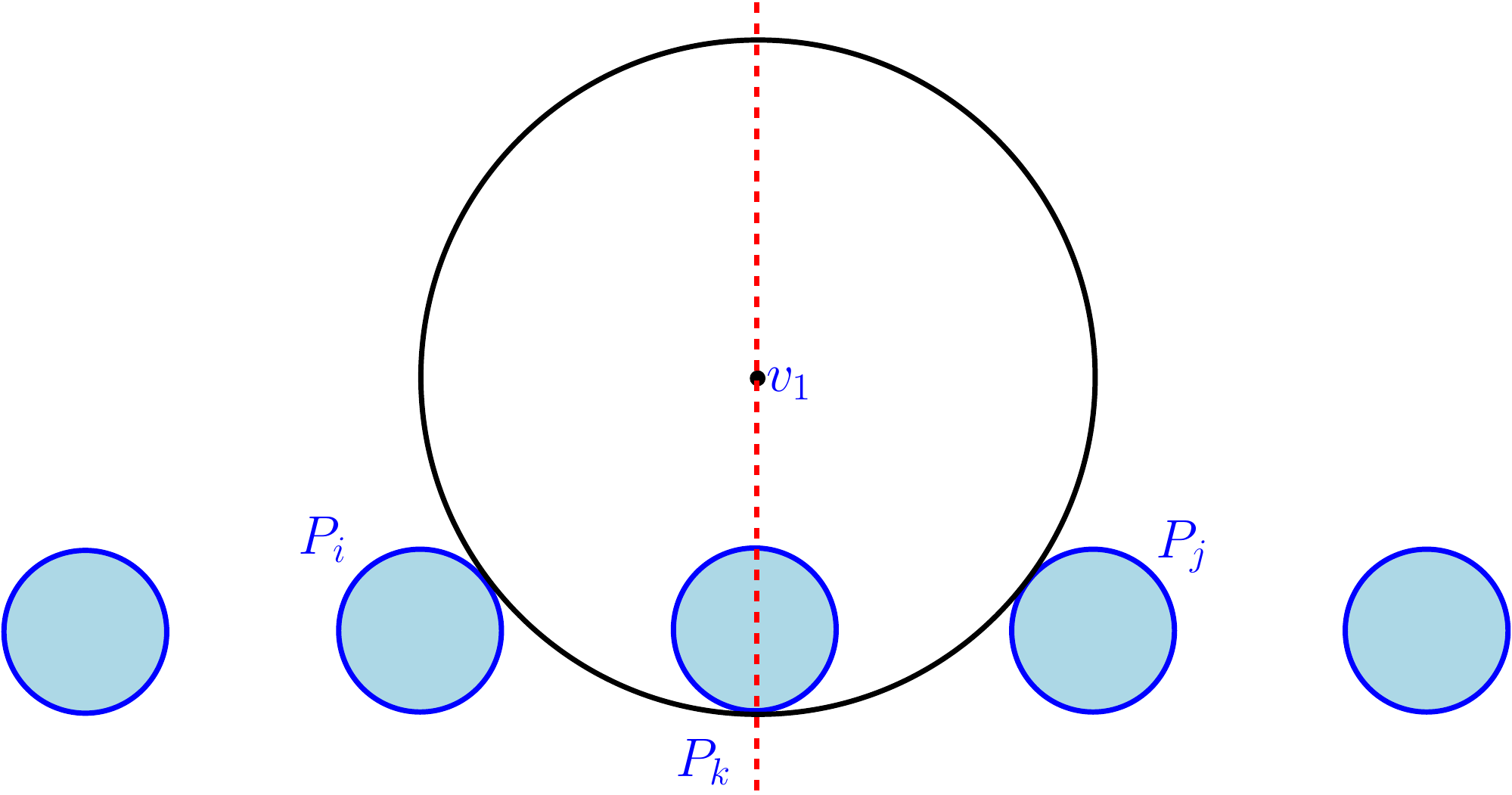}
    \caption{Any pair $(P_i, P_j)$ satisfying $j-i \geq 2$ determines
       2 vertices of $\NZVD$. Only the vertex $v_1$ is shown.}
    \figlab{l:b:d:d:s:r}
\end{figure}

\begin{proof}
    The upper bound on the complexity of $\NZVD(\P)$ follows from
    \lemref{rho}.  By the same argument as in the proof of
    \thmref{continuous1}, $\NZVD(\P)$ can be computed in
    $O(n^2\log n + \mu)$ time, where $\mu$ is the number of vertices
    in $\NZVD(\P)$.

    Next we show that there exists a set $\P$ of $n$ uncertain points
    in $\mathbb{R}^2$ such that $\NZVD(\P)$ has $\Omega(n^2)$
    vertices. Assume that $n = 2m$ for some positive integer $m$. All
    the disks $D_i$ have the same radius 1, centered at
    $c_i = (4(i-m)-2, 0)$, for $1 \leq i \leq 2m$. Any pair
    $(P_i, P_j)$ satisfying that $j - i \geq 2$ and $j+i$ is even
    determines 2 vertices: $v_1 = (2(i+j-2m-1), (j-i)^2 - 1)$, and
    $v_2 = (2(i+j-2m-1), 1-(j-i)^2)$, of $\NZVD$ (realized with $P_k$,
    $k = \frac{j+i}{2}$) (\figref{l:b:d:d:s:r}). Any pair $(P_i, P_j)$
    satisfying that $j - i \geq 2$ and $j+i$ is odd determines 2
    vertices: $v_1 = (2(i+j-2m-1), (j-i)\sqrt{(j-i)^2-4})$, and
    $v_2 = (2(i+j-2m-1), (i-j)\sqrt{(j-i)^2-4})$, of $\NZVD$ (realized
    with $P_k$, $k = \lfloor\frac{j+i}{2}\rfloor$ or
    $k = \lceil\frac{j+i}{2}\rceil$). One can verify that
    $\delta_i(v) = \delta_j(v) = \Delta_k(v) \leq \Delta_{l}(v)$, for
    $1\leq l \leq n$, $v \in \{v_1, v_2\}$. Hence, we obtain a lower
    bound of $\Omega(n^2)$ for the complexity of $\NZVD$.
\end{proof}

\paragraph{Remarks.}%
We note that the proof of \lemref{rho} is essentially a packing
argument, and therefore can be extended to the case when each
uncertainty region is a convex $\alpha$-fat semialgebraic set of
constant description complexity.  A convex set $C$ is called
\emphi{$\alpha$-fat}, if there exist two concentric disks $D$ and $D'$
so that $D \subseteq C \subseteq D'$ and the ratio between the radii
of $D'$ and $D$ is at most $\alpha$. The constant of proportionality
also depends on $\alpha$ and the description complexity of the sets
defining the uncertainty regions. This in turn implies that
$\NZVD(\P)$ has $O(\lambda n^2)$ complexity if the uncertainty regions
of $\P$ are pairwise-disjoint convex $\alpha$-fat sets, for some
constant $\alpha \ge 1$, and the ratio of the size of the largest to
the smallest region is bounded by $\lambda$.  Extension of the proof
of \lemref{rho} to this case, however, is even more technical, so we
have decided not to state this generalized result as a theorem,
especially since, in practice, a fat convex set can be approximated by
a circular disk.

\paragraph{Storing $\P_\phi$'s for $\NZVD(\P)$.}%
We store the index $i$ of each uncertain point $P_i$ instead of $P_i$
itself. If we store $\P_\phi$ for each cell $\phi$ of $\NZVD(\P)$
explicitly, the size increases by a factor of $n$. However, we observe
that for two adjacent cells $\phi$, $\phi'$ of $\NZVD(\P)$, i.e., two
cells that share a common edge, $|\P_\phi \oplus \P_{\phi'}| = 1$,
where $\oplus$ denotes the symmetric difference of two sets.
Therefore, using a persistent data structure~\cite{dsst-mdsp-89}, we
can store $\P_\phi$ for all cells of $\NZVD(\P)$ in $O(\mu)$ space,
where $\mu$ is the complexity of $\NZVD(\P)$, so that for any cell
$\phi$, $\P_\phi$ can be retrieved in $O(\log n + |\P_\phi|)$
time.\footnote{If the curves of $\Gamma$ intersect transversally at
   every vertex, it suffices to store $\P_\phi$ for each cell of
   $\NZVD(\P)$. Otherwise one may have to store $\P_\phi$ for edges
   and vertices of $\NZVD(\P)$. This does not affect the asymptotic
   performance of the data structure.}  By combining this with a
planar point-location data structure \cite{bcko-cgaa-08}, we obtain
the following:
\begin{theorem}%
    \thmlab{NZVD_continuous_query}%
    Let $\P$ be a set of $n$ uncertain points in $\Re^2$, and let
    $\mu$ be the complexity of $\NZVD(\P)$. Then, $\NZVD(\P)$ can be
    preprocessed in $O(\mu\log \mu)$ time into $\anindex$ of size
    $O(\mu)$ so that, for a query point $q\in \Re^2$, $\NZNN(q, \P)$
    can be computed in $O(\log n + t)$ time, where $t$ is the output
    size.
\end{theorem}

\subsection{Discrete case} %
\seclab{VD:discrete}
We now analyze the complexity of $\NZVD(\P)$ when the distribution of
each point $P_i$ in $\P$ is discrete. Let
$P_i = \{p_{i1}, \ldots, p_{ik}\}$. For $1\leq j \leq k$, let
$w_{ij} = \Pr[P_i \text{ is } p_{ij}]$.  As in the previous section,
for a point $x$, let
\begin{align*}
  \Delta_i(q) = \max_{1 \leq j \leq k} \dist(q, p_{ij}) \quad \mbox{and}\quad
\delta_i(q) = \min_{1 \leq j \leq k} \dist(q, p_{ij}).
\end{align*}
Note that the projection of the graph of $\Delta_i$ (resp.\
$\delta_i$) onto the $xy$-plane is the farthest-point (resp.\
nearest-point) Voronoi diagram of $P_i$. Let
$\Delta(q) = \min_{1\leq i \leq n} \Delta_i(q)$.  For each $i$, let
$\gamma_i = \{ x \in \Re^2 \mid \delta_i(x) = \Delta(x)\}$, and set
$\Gamma = \{\gamma_1, \ldots, \gamma_n\}$. Then $\NZVD(\P)$ is the
planar subdivision $\EuScript{A}(\Gamma)$ induced by $\Gamma$ (cf.\
\corref{NZVD}).

We define a few functions that will help analyze the structure of
$\NZVD(\P)$.  We first define a function
$f\colon \Re^2\times \Re^2 \rightarrow \Re$ as
\begin{equation}
    \eqlab{linear:function}%
    f(x,p) = \dist^2 (x, p) - \|x\|^2 = \|p\|^2 - 2\langle x,
    p\rangle.
\end{equation}
For $1\leq i \leq n$, define
\begin{align*}
  \varphi_i(x) = \min_{1\leq j \leq k} f(x, p_{ij}) \quad \mbox{and}\quad
  \Phi_i(x) = \max_{1\leq j \leq k} f(x, p_{ij}).
\end{align*}
Finally, we define
\[ \Phi(x) = \min_{1\le i\le n} \Phi_i(x) .\] The following lemma is
straightforward.
\begin{lemma}%
    \lemlab{linearize}%
    For any $i \le n$ and for any $q \in \reals^2$, $\delta_i(q)=r$ if
    and only if $\varphi_i(q)=r^2-\|q\|^2$.
\end{lemma}

\begin{lemma}%
    \lemlab{some:lemma}%
    For any pair $i,j$, $1 \leq i \neq j \leq n$, let
    $\gamma_{ij} = \{ x \in \Re^2 \mid \delta_i(x) = \Delta_j(x)\}$,
    then $\gamma_{ij}$ is a convex polygonal curve with $O(k)$
    vertices.
\end{lemma}
\begin{proof}
    By \lemref{linearize}, for any pair $i, j$ and for any
    $x\in\Re^2$, $\delta_i(x) = \Delta_j(x)$ if and only if
    $\varphi_i(x) = \Phi_j(x)$. Hence, $\gamma_{ij}$ is also the zero
    set of the function $\Phi_j(x) - \varphi_i(x)$.

    $\Phi_j$ is the upper envelope of $k$ linear functions, and thus
    is a piecewise-linear convex function. Similarly, $\varphi_i$, the
    lower envelope of $k$ linear functions, is a piecewise-linear
    concave function. Hence, $\Phi_j(x)-\varphi_i(x)$ is a
    piecewise-linear convex function, which implies that
    $\gamma_{ij} = \{ x \in \Re^2 \mid \Phi_j(x) = \varphi_i(x)\}$ is
    a convex polygonal curve. Since $\gamma_{ij}$ is the projection of
    the intersection curve of the graphs of $\Phi_j$ and $\varphi_i$,
    each of which is the surface of an unbounded convex polyhedron
    with at most $k$ faces, $\gamma_{ij}$ has $O(k)$ vertices.
\end{proof}

\begin{theorem}
    Let $\P = \{P_1, \ldots, P_n \}$ be a set of $n$ uncertain points
    in $\Re^2$, where each $P_i$ has a discrete distribution of size
    at most $k$. The complexity of $\NZVD(\P)$ is $\mu = O(k n^3)$ ,
    and it can be computed in $O(n^2\log n + \mu)$ expected
    time. Furthermore, it can be preprocessed in additional $O(\mu)$
    time into $\anindex$ of size $O(\mu)$ so that an $\NZNN(q)$ query
    can be answered in $O(\log \mu + t)$, where $t$ is the output
    size.
\end{theorem}

\begin{proof}
    We follow the same argument as in the proof of
    \thmref{continuous1}. We need to bound the number of intersection
    points between a pair of curves $\gamma_i$ and $\gamma_j$. Fix an
    index $u$. Let
    $\gamma_{iu} = \{ x \in \Re^2 \mid \delta_i(x) = \Delta_u(x)\}$
    and
    $\gamma_{ju} = \{ x \in \Re^2 \mid \delta_j(x) = \Delta_u(x)\}$.
    By \lemref{some:lemma}, each of $\gamma_{iu}$ and $\gamma_{ju}$ is
    a convex polygonal curve in $\Re^2$ with $O(k)$ vertices. Since
    two convex polygonal curves in general position with $n_1$ and
    $n_2$ vertices intersect in at most $n_1 + n_2$ points,
    $\gamma_{iu}$ and $\gamma_{ju}$ intersect at $O(k)$ points. Hence,
    $\gamma_i$ and $\gamma_j$ intersect at $O(kn)$ points, implying
    that $\NZVD(\P)$ has $O(k n^3)$ vertices. The running time follows
    from the proof of \thmref{continuous1}.
\end{proof}

\section{Data structures for $\NZNN$ queries}                         %
\seclab{indexing_schemes}

With the maximum size of $\NZVD$ being $\Theta(n^3)$, we present
$O(n\polylog(n))$-size data structures that circumvent the need for
constructing $\NZVD(\P)$ and answer $\NZNN$ queries in
poly-logarithmic or sublinear time. They rely on geometric data
structures for answering range-searching queries and their variants;
see \cite{a-rs-16} for a recent survey.

An $\NZNN(q)$ query is answered in two stages. The first stage
computes $\Delta(q)$, and the second stage computes all points
$P_i \in \P$ for which $\delta_i(q) < \Delta(q)$. We build a separate
$\index$ for each stage. We first describe the one for the continuous
case and then for the discrete case.

   \paragraph{Continuous case.} %
   We assume that the uncertainty region of each point $P_i$ is a disk
   $D_i$ of radius $r_i$ centered at $c_i$. Recall from \secref{NZVD}
   that the projection of the graph of the function $\Delta$ onto the
   $xy$-plane, a planar subdivision $\mathbb{M}$, is the
   (additive-weighted) Voronoi diagram of the points
   $c_1, \ldots, c_n$, and it has linear complexity. Hence
   $\mathbb{M}$ can be preprocessed in $O(n\log n)$ time into
   $\anindex$ of size $O(n)$ so that for a query point $q \in \Re^2$,
   $\Delta(q)$ can be computed in $O(\log n)$
   time~\cite{bcko-cgaa-08}.

   Next we wish to report all points $P_i \in \P$ for which
   $\delta_i(q) < \Delta(q)$, i.e., for which $D_i$ intersects the
   disk of radius $\Delta(q)$ centered at $q$. Note that the
   projection of the graph of the lower envelope of
   $\{\delta_1, \ldots, \delta_n\}$ is also an (additive-weighted)
   Voronoi diagram of the points $c_1, \ldots, c_n$ and has linear
   complexity. Recently~\cite{KMRSS} have described a data structure
   of size $O(n\polylog(n))$ that can answer the above query in
   $O(\log n + t)$ time, where $t$ is the output size.  It can be
   constructed in $O(n\polylog(n))$ randomized expected time.  We thus
   obtain the following:

\begin{theorem}%
    \thmlab{indexscheme_continuous} %
    Let $\P = \{P_1, \ldots, P_n\}$ be a set of $n$ uncertain points
    in $\Re^2$ so that the uncertainty region of each~$P_i$ is a
    disk. $\P$ can be preprocessed into $\anindex$ of size
    $O(n\polylog(n))$, so that an $\NZNN(q)$ query can be answered in
    $O(\log n + t)$ time, where $t$ is the output size.  The data
    structure can be constructed in $O(n\polylog(n))$ randomized
    expected time.
\end{theorem}

\paragraph{Remarks.}%
(i) Note that \thmref{indexscheme_continuous} gives a better result
than \thmref{NZVD_continuous_query} but the $\index$ based on
$\NZVD(\P)$ is simpler and more practical.

(ii) If we use $L_1$ or $L_\infty$ metric to compute the distance
between points and use disks in $L_1$ or $L_\infty$ metric (i.e., a
diamond or a square), then an $\NZNN(q)$ query can be answered in
$O(\log^2 n + t)$ time using $O(n\log^2 n)$ space: the first stage
remains the same and the second stage reduces to reporting a set of
axis-aligned squares that intersect a query axis-aligned
square~\cite{a-rs-16}.

\paragraph{Discrete case.} %
Next, we consider the case when each $P_i$ has a discrete distribution
of size at most $k$; set $N=nk$. The functions $\Delta_i$ and
$\delta_i$ are now more complex and thus the $\index$ for $\NZNN(q)$
queries is more involved. As in \secref{VD:discrete}, instead of
working with the functions $\delta_i$ and $\Delta_i$, we work with
$\varphi_i$ and $\Phi_i$. By \lemref{linearize}, the problem of
reporting all points $P_i$ with $\delta_i(q) < \Delta(q)$ is
equivalent to returning the points with $\varphi_i(q) < \Phi(q)$. As
for the continuous case, we construct two data structures---the first
one computes $\Phi(q)$ for a query point $q \in \reals^2$ and the
second one reports all $P_i$'s with $\varphi_i(q) < \Phi(q)$.  Note
that $\varphi_i(q) < \Phi(q)$ if and only if the point
$\widehat{q}=(q,\Phi(q)) \in \Re^3$ lies above the graph of
$\varphi_i$.  By triangulating each face of $\varphi_i$ and $\Phi_i$
if necessary, we can assume that each $\varphi_i$ is a triangulated
concave surface and each $\Phi_i$ is a triangulated convex surface.

We now describe the data structure for computing $\Phi(q)$. Note that
$\Phi(q)=\Phi_j(q)$ if $\Phi_j$ is the first surface in the set
$\{\Phi_1, \ldots, \Phi_n\}$ intersected by $\ell_q$, the vertical
line passing through $q$, in the $(+z)$-direction. We first construct
a $3$-level partition tree \cite{a-rs-16} on the set of triangles in
the graphs of $\Phi_1, \ldots, \Phi_n$, denoted by $\Sigma$, so that
the triangles of $\Sigma$ intersected by $\ell_q$, for a query point
$q\in\Re^2$, can be reported efficiently.  The partition tree stores a
family of \emph{canonical} subsets of triangles in $\Sigma$ so that
for any query point $q$, the triangles of $\Sigma$ intersected by
$\ell_q$ can be reported as the union of $O(\sqrt{N}\log^2 n)$
canonical subsets in $O(\sqrt{N}\log^2 n)$ time. Let $F_q$ denote the
family of canonical subsets reported by the query procedure. The size
of the data structure is $O(N\log^2 n)$, and it can be constructed in
$O(N\log^2 N)$ randomized expected time~\cite{a-rs-16}.  Next, for
each canonical subset $C$, let $C^*$ be the set of planes supporting
the triangles in $C$. We construct the lower envelope $L_C$ of $C^*$
(by regarding each plane in $C^*$ as the graph of a linear function),
which has size $O(|C|)$, and preprocess $L_C$ into an $O(|C|)$-size
data structure so that for a query point $q\in \reals^2$, $L_C(q)$ can
be computed in $O(\log |C|)$ time~\cite{sa-dsstg-95}. Summing over all
canonical subsets of the partition tree, the overall size of the data
structure is $O(N\log^2 N)$ and it can be constructed in
$O(N\log^3 N)$ randomized expected time.

Given a query point $q \in \Re^2$, we first query the partition tree
and compute the family $F_q$ of canonical subsets. For each canonical
set $C\in F_q$, we compute $L_C(q)$ and return the minimum among them
as $\Phi(q)$.  Since, the procedure spends $O(\log N)$ time for each
canonical subset, the overall query time is $O(\sqrt{N}\log^3 N)$.
The correctness of the procedure follows from the following
observation: $\ell_q$ intersects all triangles of a canonical subset
$C\in F_q$, so for each triangle $\tau \in C$ and its supporting plane
$\tau^*$, $\ell_q\cap\tau = \ell_q \cap \tau^*$. Therefore $L_C(q)$ is
the same as the (height of the) first intersection point of $\ell_q$
with a triangle of $C$, and $\Phi(q) = \min_{C\in F_q} L_C(q)$.

Next, we describe the data structure for reporting the points $P_i$
with $\varphi_i(q) < \Phi(q)$. It is very similar to the one just
described, except one twist.  First, as above, we construct a 3-level
partition tree on the triangles in the graphs of
$\varphi_1, \ldots, \varphi_n$. Let $C$ be a canonical subset
constructed by the partition tree, and let $C^*$ be the set of planes
supporting $C$.  Using a result by \cite{ac-ohrrt-09} (see
also~\cite{a-rs-16}), we preprocess $C^*$, in $O(|C|\log |C|)$
randomized expected time, into a data structure of size $O(|C^*|)$ so
that for a query point $\widehat{q} = (q, \Phi(q))$, all $t_C$ planes
of $C^*$ lying below $\hat{q}$ can be reported in $O(\log N+t_C)$
time.
Summing over all canonical subsets of the partition tree, the overall
size of the data structure is $O(N\log^2 N)$, and it can be
constructed in $O(N\log^3 N)$ randomized expected time.

Given a query point $q \in \Re^2$, we first query the partition tree
and compute the family $F_q$ of canonical subsets. For each canonical
set $C\in F_q$, we next report all planes of $C^*$ lying below
$\hat{q}$. The overall query time is $O(\sqrt{N}\log^3 N+t)$, where
$t$ is the output size.\footnote{We note that if $\ell_q$ passes
   through the boundary of a triangle of some $\varphi_i$, then $P_i$
   may be reported multiple times. If the points of $P_i$ are in
   general position, then the degree of each vertex of $\varphi_i$ is
   constant, so $P_i$ will be reported $O(1)$ times. However if points
   in $P_i$ are in a degenerate position, then additional care is
   needed, using standard techniques such as symbolic perturbation, to
   ensure that $P_i$ is reported only $O(1)$ times.}  The correctness
of the procedure follows from the same argument as above, namely,
since $\ell_q$ intersects all triangles of a canonical subset
$C\in F_q$, a triangle in $C$ lies below $\widehat{q}$ if and only if
the plane supporting it lies below $\widehat{q}$.

Putting everything together, we can construct, in $O(N\log^3 N)$
randomized expected time, a data structure of $O(N\log^2 N)$ size that
can answer an $\NZNN$ query in $O(\sqrt{N}\log^3 N)$ time.

Finally, we remark that the $3$-level partition tree can be replaced
by a multi-level data structure of size $O(N^2\log^2 N)$ so that the
set of triangles intersected by $\ell_q$ can be returned as the union
of $O(\log^3 N)$ canonical subsets~\cite{a-rs-16}. Using this data
structure, we can answer an $\NZNN$ query in $O(\log^4 N)$ time using
$O(N^2\log^2 N)$ space. We thus obtain the following:

\begin{theorem}
    Let $\P$ be a set of $n$ uncertain points in~$\Re^2$, each with a
    discrete distribution of size at most $k$; set $N=nk$. $\P$ can be
    preprocessed into $\anindex$ of size $O(N\log^3 N)$ so that an
    $\NZNN(q)$ query can be answered in $O(\sqrt{N}\log^3 N + t)$
    time, or into $\anindex$ of size $O(N^2\log^2 N)$ with
    $O(\log^4 N + t)$ query time, where $t$ is the output size. The
    expected preprocessing times are $O(N\log^3 N)$ and
    $O(N^2\log^3 N)$ time, respectively.
\end{theorem}

\section{Quantification probabilities}
\seclab{sec:quanprob}

We now turn our attention to the second part of answering
probabilistic NN queries, namely, returning the quantification
probabilities that are positive.  We begin with a data structure for
computing quantification probabilities exactly for the case when each
uncertain point has a discrete distribution of size at most $k$. Since
computing these quantities exactly is quite expensive and they are
small for most of the points, we focus on computing quantification
probabilities approximately.

\subsection{The exact algorithm}
\seclab{exact}

Assuming each point in $\P$ has a discrete distribution of size at
most $k$, we build the \emphi{probabilistic Voronoi diagram}
$\PVD(\P)$ that decomposes $\Re^2$ into a set of cells, so that any
point $q$ in a cell has the same $\pi_i(q)$ value for all
$P_i \in \P$; that is, for any point $q$ in this cell, we know exactly
the probability of each point $P \in \P$ being the $\NN$ of $q$.

\begin{lemma}%
    \lemlab{pvd}%
    Let $\P$ be a set of $n$ uncertain points in $\Re^2$, each with a
    discrete distribution of size at most $k$; set $N=nk$. The
    complexity of $\PVD(\P)$ is $O(N^4)$.  Moreover, there exists a
    set $\P$ of $n$ uncertain points in $\mathbb{R}^2$ with $k = 2$
    such that $\PVD(\P)$ has size $\Omega(n^4)$.
\end{lemma}

\begin{proof}
    We first prove the upper bound. There are $N$ possible
    locations. Each pair of possible locations determines a bisector,
    resulting in $O(N^2)$ bisectors. These bisectors partition the
    plane into $O(N^4)$ convex cells so that the order of all
    distances to each of the $nk$ possible locations, and thus by
    \Eqref{definition2} also all the quantification probabilities, are
    preserved within each cell. Therefore, the resulting planar
    subdivision is a refinement of $\PVD(\P)$, and thus $O(n^4k^4)$ is
    an upper bound on the complexity of $\PVD(\P)$.

    Next, we show that there exists a set $\P$ of $n$ uncertain points
    in $\mathbb{R}^2$ with $k = 2$ such that $\PVD(\P)$ has size
    $\Omega(n^4)$.  For simplicity, we describe a degenerate
    configuration of points, but the argument can be generalized to a
    non-degenerate configuration as well, by being more careful. For
    every $1 \leq i \leq n$, $P_i \in \P$ has two possible locations
    $p_{i}$ and $p'_{i}$, each with probability $0.5$.  Let $D$ be the
    unit disk centered at the origin. We choose $p_1, \ldots, p_n$
    inside $D$ so that the bisector $b_{ij}$ of every pair $p_i, p_j$,
    for $i < j$, is a distinct line and all pairs of bisectors
    intersect inside $D$. We place all $p'_i$'s at the same location
    far away from $D$, say, at $\bp = (100, 0)$. Note that the
    bisector of $\bp$ and $p_i$, for any $i \leq n$, does not
    intersect $D$, so for any point $q \in D$,
    $\dist(p_i, q) < \dist(\bp, q)$.

    \begin{figure}[t]
        \centering
        \IncludeGraphics[width=0.30\textwidth]%
        {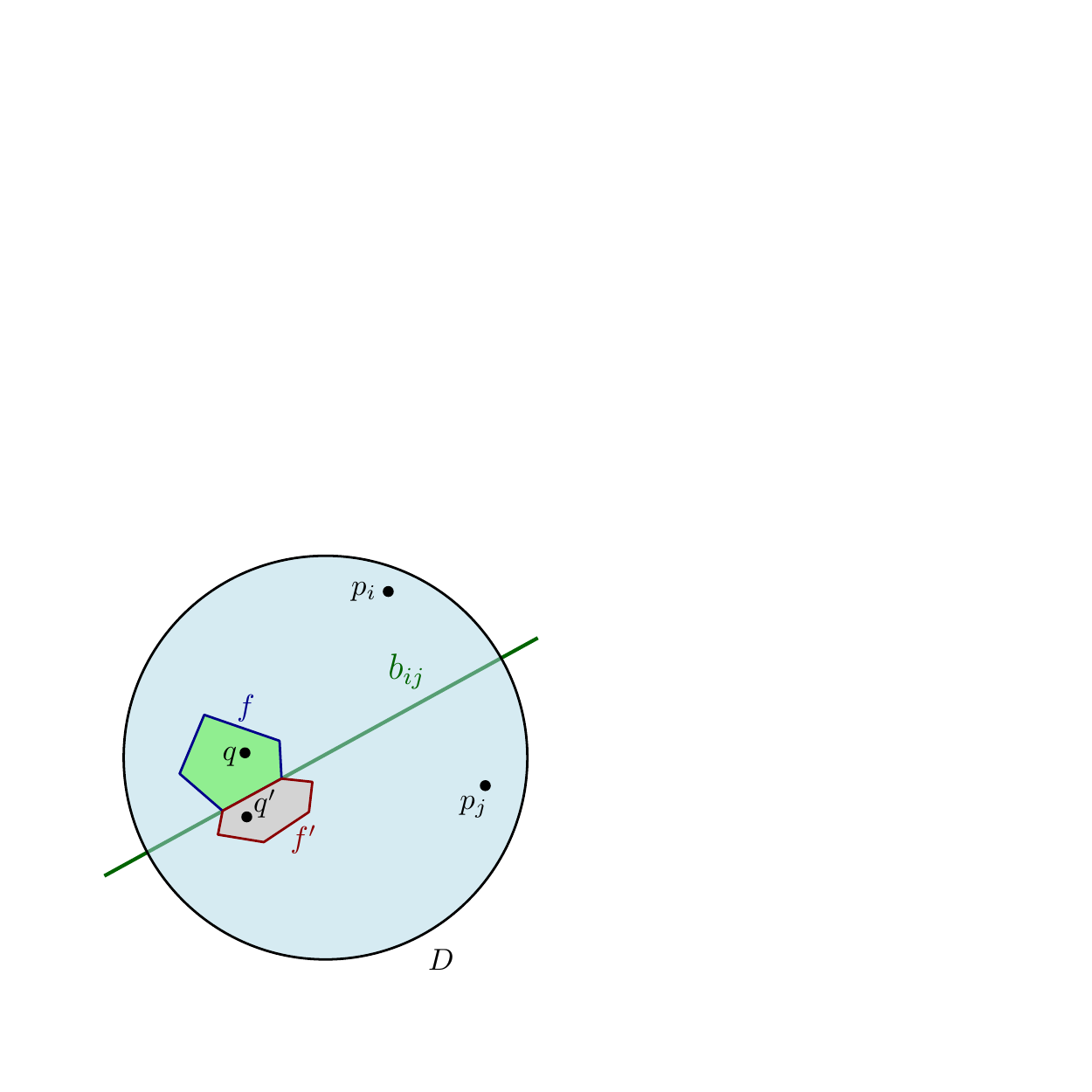}
        \caption{An illustration of the proof. Inside the unit disk
           $D$, two adjacent faces $f$ and $f'$ share a portion of the
           bisector $b_{ij}$ defined by $p_i$ and $p_j$.}
        \figlab{pvd_n_4_lower_bound}
    \end{figure}

    Let $\EuScript A$ be the arrangement of the bisectors
    $\{ b_{ij} \mid 1 \leq i < j \leq n \}$. Since all pairs of
    bisectors intersect inside $D$, $\EuScript A \cap D$ has
    $\Theta(n^4)$ faces. Let $f, f'$ be any two adjacent faces of $\A$
    inside $D$, let $b_{ij}$ be the bisector separating $f$ and $f'$,
    and let $q, q'$ be arbitrary points in the interior of $f, f'$,
    respectively.  Without loss of generality, assume that
    $\dist(p_i, q) < \dist(p_j, q)$, then
    $\dist(p_i, q') > \dist(p_j, q')$. See
    \figref{pvd_n_4_lower_bound}.  Suppose there are $r$,
    $0 \leq r < n - 1$, points of $\{p_1, \ldots, p_n\}$ that are
    closer to $q$ than $p_i$, i.e., $p_i$ (resp. $p_j$) is the
    $(r+1)$-st $\NN$ of $q$ (resp. $q'$) among $\{p_1, \ldots,
    p_n\}$. Then, by \Eqref{definition2},
    \begin{align*}
      \pi_i(q) &= 0.5\cdot (1 - 0.5)^{r} + 0.5\cdot (1 - 0.5)^{n-1}
                 =
                 0.5^{r+1} + 0.5^n, \text{ and}\\
      \pi_j(q) &= 0.5\cdot (1 - 0.5)^{r+1} + 0.5\cdot (1 -
                 0.5)^{n-1} = 0.5^{r+2} + 0.5^n.
    \end{align*}
    Symmetrically, $\pi_i(q') = 0.5^{r+2} + 0.5^n$ and
    $\pi_j(q') = 0.5^{r+1} + 0.5^n$. In particular
    $\pi_i(q) \neq \pi_i(q')$ and $\pi_j(q) \neq \pi_j(q')$.  In other
    words, any two adjacent faces of $\EuScript A$ inside $D$ have
    distinct quantification probability vectors, implying that
    $\PVD(\P)$ has $\Omega(n^4)$ complexity.
\end{proof}

As in \secref{VD:continuous}, we can store the quantification
probabilities for all faces of $\PVD(\P)$ by using $O(1)$ storage per
face. Hence, by preprocessing $\PVD(\P)$ for point-location queries,
for a query point $q$, we can report all $t$ quantification
probabilities that are positive in $O(\log N+t)$ time.

\begin{theorem}
    \thmlab{qp-query} Let $\P$ be a set of $n$ uncertain points in
    $\Re^2$, each with a discrete distribution of size at most $k$;
    set $N=nk$. $\P$ can be preprocessed in time $O(N^4\log N)$ time
    into a data structure of size $O(N^4)$ that can report all $t$
    positive quantification probabilities of a query point in time
    $O(\log N+t)$.
\end{theorem}


\subsection{A Monte-Carlo algorithm} %
\seclab{Monte:Carlo}

In this section we describe a simple Monte-Carlo approach to build
$\anindex$ for quickly computing $ \wpi_i(q)$ for all $P_i$ for any
query point $q$, which approximates the quantification probability
$\pi_i(q)$.  For a fixed value $s$, to be specified later, the
preprocessing step works in $s$ rounds.  In the $j$-\th round the
algorithm creates a sample
$R_j = \{r_{j1}, r_{j2}, \ldots, r_{jn}\}\subseteq \Re^2$ by choosing
each $r_{ji}$ using the distribution of $P_i$.  For each
$j \in \{1, \ldots, s\}$, we construct the Voronoi diagram $\Vor(R_j)$
in $O(n \log n)$ time and preprocess it for point-location queries in
additional $O(n \log n)$ time.

To estimate quantification probabilities of a query~$q$, we initialize
a counter $c_i =0$ for each point $P_i$. For each
$j \in \{1, \ldots, s\}$, we find the point $r_{ji}\in R_j$ whose cell
in $\Vor(R_j)$ contains the query point $q$, and increment $c_i$ by
$1$.  Finally we estimate $\wpi_i(q) = c_i/s$.  Note that at most $s$
distinct $c_i$'s have nonzero values, so we can implicitly set the
remaining $\wpi_i(q)$'s to $0$.

\paragraph{Discrete case.}%
If each $P_i \in \P$ has a discrete distribution of size $k$, then
this algorithm can be implemented very efficiently.  Each $r_{ji}$ can
be selected in $O(\log k)$ time after preprocessing each $P_i$, in
$O(k)$ time, into a balanced binary tree~\cite{mr-ra-95}.  Thus, total
preprocessing takes
$O(s (n (\log n + \log k)) + nk) = O(nk+ sn\log (nk))$ time and
$O(sn)$ space, and each query takes $O(s \log n)$ time.

It remains to determine the value of $s$ so that
$|\pi_i(q) - \wpi_i(q)| \leq \eps$ for all $P_i$ and all queries $q$,
with probability at least $1-\delta$.  For fixed $q$, $P_i$, and
instantiation $R_j$, let $\ranX_{ji}$ be the random indicator
variable, which is 1 if $r_{ji}$ is the $\NN$ of $q$ and 0 otherwise.
Since $\E[\ranX_{ji}] = \pi_i(q)$ and $\ranX_i \in \{0,1\}$, applying
the Chernoff-Hoeffding bound~\cite{mr-ra-95} to
   \begin{align*}
     \wpi_i(q) = \frac{c_i}{s} = \frac{1}{s} \sum_{j=1}^s \ranX_{ji},
   \end{align*}
 we obtain that
\begin{align}
  \eqlab{chernoff} \Pr\left[\strut| \wpi_i(q) - \pi_i(q)| \geq
  \eps\right] \leq 2 \exp(-2\eps^2 s).
\end{align}

For each cell of $\PVD(\P)$, we choose one point, and let $Q$ be the
resulting set of points. If $| \wpi_i(q) - \pi_i(q)| \leq \eps$ for
every point $q \in Q$, then $| \wpi_i(q) - \pi_i(q)| \leq \eps$ for
every point $q \in \Re^2$. Since there are $n$ different values of
$i$, by applying the union bound to \Eqref{chernoff}, the probability
that there exist a point $q\in\Re^2$ and an index
$i \in \{1, \ldots, n\}$ with $| \wpi_i(q) - \pi_i(q)| \geq \eps$ is
at most $2n|Q|\exp(-2\eps^2s)$. Hence, by setting
\begin{align*}
  s = \frac{1}{2\eps^2}\ln \frac{2n|Q|}{\delta},
\end{align*}
$| \wpi_i(q) - \pi_i(q)| \leq \eps$ for all $q \in \Re^2$ and for all
$i \in \{1, \ldots, n\}$, with probability at least $1 - \delta$.
By~\lemref{pvd}, $|Q| = O(n^4k^4)$, so we obtain the following result.
\begin{theorem}
    Let $\P$ be a set of $n$ uncertain points in~$\Re^2$, each with a
    discrete distribution of size $k$, and let
    $\eps,\delta \in (0, 1)$ be two parameters. $\P$ can be
    preprocessed, in
    \[O(nk + (n/\eps^2) \log(nk) \log(nk/\delta))\] time, into
    $\anindex$ of size $O((n/\eps^2) \log(nk/\delta))$, which
    computes, for any query point $q \in \Re^2$, in
    $O((1/\eps^2) \log(nk/\delta) \log n)$ time, a value $ \wpi_i(q)$
    for every $P_i$ such that $|\pi_i(q) - \wpi_i(q)| \leq \eps$ for
    all $i$ with probability at least $1-\delta$.
\end{theorem}

\paragraph{Continuous case.} %
There are two technical issues in extending this technique and
analysis to continuous distributions.  First, how we instantiate a
certain point $r_i$ from each $P_i$.  Herein we assume the
representation of the \pdf is such that this can be done in constant
time for each $P_i$.

Second, we need to bound the number of distinct queries that need to
be considered to apply the union bound as we did above.  Since
$\pi_i(q)$ may vary continuously with the query location, unlike the
discrete case, we cannot hope for a bounded number of distinct
results.  However, we just need to define a finite set $\QQ$ of query
points so that for any point $q \in \Re^2$, there is a point
$q' \in \QQ$ such that $\max_i |\pi_i(q) - \pi_i(q')| \leq \eps/2$.
Then, we can choose $s$ large enough so that it permits at most
$\eps/2$ error on each query in $\QQ$.  Specifically, choosing
$s = O((1/\eps^2) \log (n |\QQ|/\delta))$ is sufficient, so all that
remains is to bound $|\QQ|$.

To choose $\QQ$, we show that each \pdf of $P_i$ can be approximated
with a discrete distribution of size $O((n^2/\eps^2)\log(n/\delta))$,
and then reduce the problem to the discrete case.

For parameters $\alpha > 0$ and $\delta' \in (0,1)$, set
\begin{align*}
  k(\alpha) = \frac{c}{\alpha^2}\log\frac{1}{\delta'},
\end{align*}
where $c$ is a constant. For each $i \in \{1, \ldots, n\}$, we choose
a random sample $\barX{P}_i \subset P_i$ of size $k(\alpha)$,
according to the distribution defined by the location $\pdf$ $f_i$ of
$P_i$. We regard $\barX{P}_i$ as an uncertain point with uniform
location probability. Set
$\barX{\P} = \brc{\barX{P}_1, \ldots, \barX{P}_n }$.

For a point $q \in \Re^2$, let $\barX{G}_{q, i}$ denote the $\cdf$ of
the distance between $q$ and $\barX{P}_i$, i.e.,
$\barX{G}_{q, i}(r) = \Pr[\dist(q, \barX{P}_i) \leq r]$, or
equivalently, it is the probability of $\barX{P}_i$ lying in the disk
of radius $r$ centered at $q$.  A well-known result in the theory of
random sampling~\cite{lls-ibscl-01,vc-ucrfe-71} implies that for all
$q \in \Re^2$ and $r \geq 0$,
\begin{align}
  \eqlab{sampling:theory} \left|G_{q,i}(r) -
  \barX{G}_{q,i}(r)\right| \leq \alpha,
\end{align}
with probability at least $1-\delta'$, provided that the constant $c$
in $k(\alpha)$ is chosen sufficiently large.

Let $\barX{\pi}_i(q)$ denote the probability of $\barX{P}_i$ being the
$\NN$ of $q$ in $\barX{\P}$.  We prove the following:

\begin{lemma}
    For any $q \in \Re^2$ and for any fixed $i \in \{1, \ldots, n\}$,
    \begin{align*}
      |\pi_i(q) - \barX{\pi}_i(q)| \leq \alpha n,
    \end{align*}
    with probability at least $1 - \delta'$.
\end{lemma}

\begin{proof}
    Recall that by \Eqref{definition},
    \begin{equation*}
        \pi_i(q) = \int_0^{\infty} g_{q, i}(r) \prod_{j\neq i} (1 -
        G_{q,j}(r)) \dir r.
    \end{equation*}
    Using \Eqref{sampling:theory}, and the fact that
    $G_{q, j}(r), \barX{G}_{q, j}(r)\in [0, 1]$ for all $j$, we obtain
    \begin{equation*}
        \pi_i(q) \leq \int_0^{\infty} g_{q, i}(r) \prod_{j\neq i} (1 -
        \barX{G}_{q,j}(r)) \dir r + (n-1)\alpha.
    \end{equation*}

    Note that $\prod_{j\neq i} (1-\barX{G}_{q,j}(r))$ is the
    probability that the closest point of $q$ in
    $\barX{\P}\setminus\brc{\barX{P}_i}$ is at least distance $r$ away
    from~$q$. Let $h_{q, i}$ be the $\pdf$ of the distance between $q$
    and its closest point in $\barX{\P} \setminus \brc{
       \barX{P}_i}$. Then
    \begin{equation*}
        \prod_{j\neq i} (1-\barX{G}_{q,j}(r)) = \int_r^\infty h_{q,
           i}(\theta)\dir\theta.
    \end{equation*}

    Therefore
    \begin{equation*}
        \pi_i(q) \leq \int_0^{\infty}\int_r^\infty g_{q, i}(r) h_{q,
           i}(\theta)\dir\theta\dir r + (n-1)\alpha.
    \end{equation*}
    By reversing the order of integration, we obtain
    \begin{align*}
      \pi_i(q) %
      & \leq%
        \int_0^{\infty}\int_0^\theta h_{q, i}(\theta)g_{q, i}(r)\dir
        r\dir\theta + (n-1)\alpha%
        = %
        \int_0^{\infty} h_{q, i}(\theta) G_{q, i}(\theta)\dir\theta +
        (n-1)\alpha%
      \\&%
          \leq%
          \int_0^{\infty} h_{q, i}(\theta) (\barX{G}_{q, i}(\theta) +
          \alpha)\dir\theta + (n-1)\alpha
      \qquad \text{(using \Eqref{sampling:theory})}\\
      & =%
        \int_0^{\infty} h_{q, i}(\theta) \barX{G}_{q,
        i}(\theta)\dir\theta + n\alpha
        =%
        \barX{\pi}_i(q) + n\alpha.
    \end{align*}

    A similar argument shows that
    $ \pi_i(q) \geq \barX{\pi}_i(q) - n\alpha$.  This completes the
    proof of the lemma.
\end{proof}

Thus, by setting $\alpha = \eps/(2n)$, a random sample $\barX{P}_i$ of
size $O((n^2/\eps^2) \log(n/\delta))$ from each $P_i$ ensures that
\begin{align}
  \eqlab{continous} |\pi_i(q) - \barX{\pi}_i(q)| \leq \eps/2
\end{align}
for all queries. By choosing $\delta' = \delta/(2n)$,
\Eqref{continous} holds for all $i \in \{1, \ldots, n\}$ with
probability at least $1 - \delta/2$.

We consider $\PVD(\barX{\P})$, choose one point from each of its
cells, and set $\QQ$ to be the resulting set of points. For a point
$q \in \Re^2$, let $\qq \in \QQ$ be the representative point of the
cell of $\PVD(\barX{\P})$ that contains $q$. Then,
$|\pi_i(q) - \barX{\pi}_i(\qq)| < \eps/2$ for all points $q \in \Re^2$
and $i \in \{1, \ldots, n\}$, with probability at least
$1 - \delta/2$.

Now applying the analysis for the discrete case to the point set
$\barX{\P}$, if we choose
\begin{align*}
  s = O\Bigl(\frac{1}{\eps^2} \log \frac{n|\QQ|}{\delta}\Bigr),
\end{align*}
then $\cardin{\barX{\pi}_i(q) - \wpi_i(q)} < \eps$ for all points
$q \in \Re^2$ and for all $i \in \{1, \ldots, n\}$ with probability at
least $1 - \delta/2$.  Since
\begin{align*}
  \cardin{\barX{P}_i}%
  =%
  k\Bigl(\frac{\eps}{2n}\Bigr) = O\Bigl(\frac{n^2}{\eps^2}\log
  \frac{n}{\delta}\Bigr),
\end{align*}
by~\lemref{pvd},
\begin{align*}
  |\QQ| = O\Bigl(n^4
  \Bigl(k\Bigl(\frac{\eps}{2n}\Bigr)\Bigr)^4\Bigr) =
  O\Bigl(\frac{n^{12}}{\eps^8} \log^4 \frac{n}{\delta}\Bigr).
\end{align*}

Putting everything together, we obtain the following.

\begin{theorem}
    Let $\P= \{P_1, \ldots, P_n\}$ be a set of $n$ uncertain points in
    $\Re^2$ so a random instantiation of $P_i$ can be performed in
    $O(1)$ time, and let $\eps,\delta \in (0,1)$ be two parameters.
    $\P$ can be preprocessed in
    $O((n/\eps^2) \log (n/\eps\delta) \log n)$ time into $\anindex$ of
    size $O((n/\eps^2) \log(n/\eps\delta))$ that computes for any
    query point $q \in \Re^2$, in
    $O((1/\eps^2) \log(n/\eps\delta) \log n)$ time, a value
    $\wpi_i(q)$ for every $P_i$ such that
    $|\pi_i(q) - \wpi_i(q)| \leq \eps$ for all $i$ with probability at
    least $1-\delta$.
\end{theorem}

\subsection{Spiral search algorithm}
\seclab{SpiralSearch}

If the distribution of each point in $\P$ is discrete, then there is
an alternative approach to approximate the quantification
probabilities for a given query $q$: set a parameter $m > 1$, choose
the $m$ points of $S = \bigcup_{i = 1}^n P_i$ that are closest to $q$,
and use only these $m$ points to estimate $\pi_i(q)$ for each
$P_i$. We show that this works for a small value of $m$ when, for each
$P_i$, each location is approximately equally likely, but is not
efficient if the location probabilities vary significantly.

   Recall that $w_{ij}$ is the location probability of a point
   $p_{ij} \in P_i$.  Set $S = \bigcup_{i = 1}^n P_i$ to be the set of
   all possible locations of points in $\P$.  We define the quantity
   \begin{align}
     \eqlab{max:ratio} \rho = \frac{\max_{i, j} w_{ij}}{\min_{i, j}
     w_{ij}},
   \end{align}
   the ratio of the largest to the smallest location probability over
   all points of $S$, as the \emphi{spread} of location
   probabilities. Set
   \begin{align*}
     m(\rho, \eps) = \rho k\ln (1/\eps) + k-1.
   \end{align*}

   Fix a query point $q \in \Re^2$. Let $\barX{S} \subseteq S$ be the
   $m(\rho, \eps)$ nearest neighbors of $q$ in $S$,
   $\barX{P}_{i} = \barX{S} \cap P_i$, and
   $\barX{\P} = \{\barX{P}_{1}, \ldots, \barX{P}_{n}\}.$ Note that
   $\barX{w}_i = \sum_{p_{i, a} \in \barX{P}_{i}} w_{i, a}$ is not
   necessarily equal to 1, so we cannot regard $\barX{P}_i$ as an
   uncertain point in our model, but still it will be useful to think
   of $\barX{P}_{i}$ as an uncertain point that does not exist with
   probability $1 - \barX{w}_i$.

   For a set $Y$ of points and another point $\xi \in \Re^2$, let
   \begin{align*}
     Y[\xi] = \Set{\bigl. p \in Y}{ \dist(q, p) \leq \dist(q, \xi)}.
   \end{align*}
   For a point $p := p_{i, a} \in P_i$, the probability that $p$ is
   the $\NN$ of $q$ in $\P$, denoted by $\eta(p; q)$, is
   \begin{align}
     \eqlab{prob}%
     \eta(p; q) = w_{i,a} \prod_{j\neq i} \Bigl(1 - \sum_{p_{j, \ell} \in
     P_j[p]} w_{j, \ell} \Bigr).
   \end{align}

   Moreover,
   \begin{align}
     \eqlab{prb:2}%
     \pi_i(q) = \sum_{p_{i,a} \in P_i} \eta(p_{i,a}; q).
   \end{align}

   \def\weta{\widehat{\eta}}

   For each $i \leq n$, $q \in \Re^2$, and $p_{i,a}\in \barX{P}_i$, we
   analogously define the quantities $\weta(p_{i,a};q)$ and
   $\wpi_{i}(q)$ using \Eqref{prob} and \Eqref{prb:2} but replacing
   $P_{j}$ with $\barX{P}_{j}$ for every $j \in \{1, \ldots, n\}$.
   Intuitively, if $\barX{\P}$ were a family of uncertain points, then
   $\wpi_i(q)$ would be the probability of $\barX{P}_{i}$ being the
   $\NN$ of $q$ in $\barX{\P}$.

\begin{lemma}
    For all $i \in \{1, \ldots, n\}$,
    \begin{equation*}
        \wpi_i(q) \le  \pi_i(q) \le  \wpi_i(q) + \eps.
    \end{equation*}
\end{lemma}

\begin{proof}
    Fix a point $p \in P_i$.  If $p \in \barX{P}_i$, then for all
    $j\ne i$, $\barX{P}_j[p] = P_j[p]$, therefore by \Eqref{prob},
    $\eta(p; q) = \weta(p;q)$.

    Hence, by \Eqref{prb:2},
    \begin{align}
      \pi_i(q) 
      &=%
        \sum_{p \in \barX{P}_i} \eta(p;q) + 
        \sum_{p \in P_i\setminus\barX{P}_i} \eta(p; q) %
        =%
          \sum_{p \in \barX{P}_i} \weta(p;q) + 
              \sum_{p \in P_i\setminus\barX{P}_i} \eta(p; q)
      = \wpi_i(q) +
        \sum_{p \in P_i\setminus\barX{P}_i} \eta(p; q) \eqlab{pi-bound}.
    \end{align}
    Therefore $\wpi_i(q) \le \pi_i(q)$. Next, we bound the second term
    in the right hand side of \Eqref{pi-bound}.  Let
    $p \in P_i \setminus \barX{P}_i$.  Set $x_j = |P_j[p]|$, for
    $j \ne i$, and $m' = \sum_{j\neq i} x_j$. Since
    $P_i\setminus\barX{P}_i \ne \emptyset$, $|\barX{P}_i| \le k-1$ and
    $m' = |\barX{S}|-|\barX{P}_i| \ge \rho k\ln(1/\eps)$.  Note that
    each $w_{j,a} \ge 1/\rho k$. Therefore,
    \begin{align*}%
      \eta(p;q)%
      &=%
        w_{i,a} \prod_{\substack{j \neq i}} \Bigl(1 -
        \sum_{p_\ell \in P_{j}[p]}w_{j,\ell}\Bigr) 
        \leq%
        w_{i,a} \prod_{\substack{j \neq i}} \Bigl(1 -
        \frac{x_j}{\rho k} \Bigr)%
        \leq
        w_{i,a} \prod_{\substack{j \neq i}} \exp \left
        (-x_j/\rho k \right ) 
        \\&
            = w_{i,a} \exp \left (-m'/\rho k \right) \le w_{i,a}\eps.
    \end{align*}
    Consequently,
    \begin{equation}
        \eqlab{tail-bound}
        \sum_{p \in  P_i \setminus \barX{P}_i} \eta(p; q) \le 
        \sum_{p \in  P_i \setminus \barX{P}_i} \eps w_{i,a} \le \eps. 
    \end{equation}
    Plugging \Eqref{tail-bound} in \Eqref{pi-bound}, we obtain
    $\pi_i(q) \le \wpi_i(q)+\eps$, as claimed. This completes the
    proof of the lemma.
\end{proof}

For any $i$, if $P_i\cap \barX{S}(q) = \emptyset$, then we can
implicitly set $\wpi_i(q)$ to $0$.  Using the data structure by
\cite{ac-ohrrt-09}, $S$ can be preprocessed in $O(N\log N)$ randomized
expected time into a data structure of $O(N)$ size so that
$m := m(\rho,\eps)$ nearest neighbors of a query point can be reported
in $O(m + \log N)$ time.  We thus obtain the following result.

\begin{theorem}
    Let $\P$ be a set of $n$ uncertain points in $\Re^2$, each with a
    discrete distribution of size at most $k$, let $\rho$ be the
    spread of the location probabilities, and let $N=nk$.  $\P$ can be
    preprocessed in $O(N \log N)$ expected time into $\anindex$ of
    size $O(N)$, so that for a query point $q \in \Re^2$ and a
    parameter $\eps \in (0, 1)$, it can compute, in time
    $O(\rho k \log(1/\eps) + \log N)$, values $\wpi_i(q)$ for all
    $P_i \in \P$ such that $| \pi_i(q) - \wpi_i(q) | \leq \eps$ for
    all $i \in \{1, \ldots, n\}$.
\end{theorem}

\paragraph{Remarks.} %
(i) This approach is not efficient when the spread of location
probabilities is unbounded.  In this case, one may have to retrieve
$\Omega(n)$ points. Another approach may be to ignore points with
weight smaller than $\eps/k$, since even $k$ such weights from a
single uncertain point $P_i$ cannot contribute more than $\eps$ to
$\pi_i(q)$.  However, the union of all such points may distort other
probabilities.

Consider the following example.  Let $p_1 \in P_1 \in \P$ be the
closest point to the query point $q$. Let $w(p_1) = 3\eps$. Let the
next $n/2$ closest points $p_3, \ldots, p_{n/2+2}$ be from different
uncertain points $P_3, \ldots, P_{n/2+2}$ and each have weights
$w(p) = 2/n \ll \eps/k$. Let the next closest point
$p_2 \in P_2 \in \P$ have weight $w(p_2) = 5\eps $. With probability
$\pi_{p_1}(q) = 3\eps$ the nearest neighbor is~$p_1$.  The probability
that $p_2$ is the nearest neighbor is
$\pi_{p_2}(q) = (5\eps)(1-3\eps)(1-2/n)^{n/2} < (5\eps)(1 -
3\eps)(1/e) < 2\eps$.  Thus, $p_1$ is more likely to be the nearest
neighbor than $p_2$. However, if we ignore points
$p_3, \ldots, p_{n/2+2}$ because they have small weights, then we
calculate $p_2$ has probability
$\wpi_{p_2}(q) = (1 - 3\eps)(5\eps) > 4\eps$ for being the nearest
neighbor (assuming that $\eps$ is small enough).  So $\pi_2(q)$ will
be off by more than $2 \eps$ and it would incorrectly appear that
$p_2$ is more likely to be the nearest neighbor than~$p_1$.

(ii) Though the the data structure by \cite{ac-ohrrt-09} is optimal
theoretically, it is too complex to be implemented.  Instead, one may
use the order-$m$ Voronoi diagram to retrieve the $m$ closest points
(in unsorted order) to $q$.  This would yield $\anindex$ with
$O(m(nk-m))$ space and $O(m(nk - m)\log(nk) + nk\log^3(nk))$ expected
preprocessing time~\cite{abms-claho-98}, while preserving the query
time $O(\log(nk) + m)$, where $m = O(\rho
k\log(\rho/\eps))$. Alternatively, one may use quad-trees and a
branch-and-bound algorithm to retrieve $m$ points of $S$ closest to
$q$~\cite{h-gaa-11}.

\section{Conclusions}
In this paper, we investigated $\NN$ queries in a probabilistic
framework in which the location of each input point is specified as a
probability distribution function. We presented efficient methods for
returning all points with non-zero probability of being the nearest
neighbor, estimating the quantification probabilities and using it for
threshold $\NN$ queries. We conclude by mentioning two open problems:
\smallskip%
\begin{compactenum}[(i)]
    \item The lower-bound constructions for the complexity of
    $\NZVD(\P)$ are created very carefully, and these configurations
    are unlikely to occur in practice.  A natural question is to
    characterize the sets of uncertain points for which the complexity
    of $\NZVD(\P)$ is near linear.
    
    \item Are there simple and practical linear-size data structures
    for answering $\NZNN$ queries in sublinear time?
    
    \item Can we extend the spiral search method to continuous
    distributions (at least for some simple, well-behaved
    distributions, such as Gaussian), so that the query time is always
    sublinear?
\end{compactenum}

\hypersetup{allcolors=black}%
\newcommand{\etalchar}[1]{$^{#1}$}
 \providecommand{\CNFX}[1]{ {\em{\textrm{(#1)}}}}
  \providecommand{\tildegen}{{\protect\raisebox{-0.1cm}{\symbol{'176}\hspace{-0.03cm}}}}
  \providecommand{\SarielWWWPapersAddr}{http://sarielhp.org/p/}
  \providecommand{\SarielWWWPapers}{http://sarielhp.org/p/}
  \providecommand{\urlSarielPaper}[1]{\href{\SarielWWWPapersAddr/#1}{\SarielWWWPapers{}/#1}}
  \providecommand{\Badoiu}{B\u{a}doiu}
  \providecommand{\Barany}{B{\'a}r{\'a}ny}
  \providecommand{\Bronimman}{Br{\"o}nnimann}  \providecommand{\Erdos}{Erd{\H
  o}s}  \providecommand{\Gartner}{G{\"a}rtner}
  \providecommand{\Matousek}{Matou{\v s}ek}
  \providecommand{\Merigot}{M{\'{}e}rigot}
  \providecommand{\CNFSoCG}{\CNFX{SoCG}}
  \providecommand{\CNFCCCG}{\CNFX{CCCG}}
  \providecommand{\CNFFOCS}{\CNFX{FOCS}}
  \providecommand{\CNFSODA}{\CNFX{SODA}}
  \providecommand{\CNFSTOC}{\CNFX{STOC}}
  \providecommand{\CNFBROADNETS}{\CNFX{BROADNETS}}
  \providecommand{\CNFESA}{\CNFX{ESA}}
  \providecommand{\CNFFSTTCS}{\CNFX{FSTTCS}}
  \providecommand{\CNFIJCAI}{\CNFX{IJCAI}}
  \providecommand{\CNFINFOCOM}{\CNFX{INFOCOM}}
  \providecommand{\CNFIPCO}{\CNFX{IPCO}}
  \providecommand{\CNFISAAC}{\CNFX{ISAAC}}
  \providecommand{\CNFLICS}{\CNFX{LICS}}
  \providecommand{\CNFPODS}{\CNFX{PODS}}
  \providecommand{\CNFSWAT}{\CNFX{SWAT}}
  \providecommand{\CNFWADS}{\CNFX{WADS}}


\begin{thebibliography}{dBCKO08}

\bibitem[AB86]{ab-gdt-86}
P.~F. Ash and E.~D. Bolker.
\newblock  Generalized {D}irichlet tessellations.
\newblock {\em Geometriae Dedicata}, 20:209--243, 1986.

\bibitem[ABMS98]{abms-claho-98}
\href{http://www.cs.duke.edu/~pankaj}{P.~K.~{Agarwal}}, {M. de} Berg, \href{http://kam.mff.cuni.cz/~matousek}{J. Matou{\v s}ek}, and \href{http://www.win.tue.nl/~ocheong/}{O.~{Schwarzkopf}}.
\newblock  Constructing levels in arrangements and higher order {Voronoi}
  diagrams.
\newblock {\em SIAM J. Comput.}, 27:654--667, 1998.

\bibitem[AC09]{ac-ohrrt-09}
P.~Afshani and \href{http://www.math.uwaterloo.ca/~tmchan/}{T.~M.~{Chan}}.
\newblock  Optimal halfspace range reporting in three dimensions.
\newblock In {\em Proc. 20th ACM-SIAM Sympos. Discrete Algs.}, pages 180--186,
  2009.

\bibitem[AESZ12]{aesz-nnsuu-12}
\href{http://www.cs.duke.edu/~pankaj}{P.~K.~{Agarwal}}, \href{http://www.cs.arizona.edu/~alon/}{A.~{Efrat}}, S.~Sankararaman, and W.~Zhang.
\newblock  Nearest-neighbor searching under uncertainty.
\newblock In {\em Proc. 31st ACM Sympos. Principles Database Syst.}, pages
  225--236, 2012.

\bibitem[Aga16]{a-rs-16}
\href{http://www.cs.duke.edu/~pankaj}{P.~K.~{Agarwal}}.
\newblock  Range searching.
\newblock In J.~E. Goodman, \href{http://cs.smith.edu/~orourke/}{J.~{O'Rourke}}, and C.~Toth, editors, {\em Handbook of
  Discrete and Computational Geometry}, chapter~36, page to apepar. CRC Press
  LLC, 3rd edition, 2016.

\bibitem[Agg09]{a-mmud-09}
C.~C. Aggarwal.
\newblock  {\em Managing and Mining Uncertain Data}.
\newblock Springer-Verlag, 2009.

\bibitem[AS00]{as-aa-00}
\href{http://www.cs.duke.edu/~pankaj}{P.~K.~{Agarwal}} and \href{http://www.math.tau.ac.il/~michas}{M.~{Sharir}}.
\newblock  Arrangements and their applications.
\newblock In J.-R. Sack and J.~Urrutia, editors, {\em Handbook of Computational
  Geometry}, pages 49--119. North-Holland Publishing Co., Amsterdam, 2000.

\bibitem[BEK{\etalchar{+}}11]{bekmr-nppas-11}
T.~Bernecker, T.~Emrich, H.-P. Kriegel, N.~Mamoulis, M.~Renz, and A.~Zuefle.
\newblock  A novel probabilistic pruning approach to speed up similarity
  queries in uncertain databases.
\newblock In {\em Proc. 27th IEEE Int. Conf. Data Eng.}, pages 339--350, 2011.

\bibitem[BSI08]{bsi-estkp-08}
G.~Beskales, M.~A. Soliman, and I.~F. IIyas.
\newblock  Efficient search for the top-k probable nearest neighbors in
  uncertain databases.
\newblock {\em Proc. VLDB Endow.}, 1(1):326--339, 2008.

\bibitem[CCCX09]{cccx-eptkn-2009}
R.~Cheng, L.~Chen, J.~Chen, and X.~Xie.
\newblock  Evaluating probability threshold $k$-nearest-neighbor queries over
  uncertain data.
\newblock In {\em Proc. 12th Int. Conf. Ext. Database Tech.}, pages 672--683,
  2009.

\bibitem[CCMC08]{ccmc-pvecn-08}
R.~Cheng, J.~Chen, M.~Mokbel, and C.~Chow.
\newblock  Probabilistic verifiers: Evaluating constrained nearest-neighbor
  queries over uncertain data.
\newblock In {\em Proc. 24th IEEE Int. Conf. Data Eng.}, pages 973--982, 2008.

\bibitem[CKP04]{ckp-qidmo-04}
R.~Cheng, D.~V. Kalashnikov, and S.~Prabhakar.
\newblock  Querying imprecise data in moving object environments.
\newblock {\em IEEE Trans. Know. Data Eng.}, 16(9):1112--1127, 2004.

\bibitem[CXY{\etalchar{+}}10]{cxycs-uvdvd-10}
R.~Cheng, X.~Xie, M.~L. Yiu, J.~Chen, and L.~Sun.
\newblock  Uv-diagram: A {Voronoi} diagram for uncertain data.
\newblock In {\em Proc. 26th IEEE Int. Conf. Data Eng.}, pages 796--807, 2010.

\bibitem[dBCKO08]{bcko-cgaa-08}
\href{http://www.win.tue.nl/~mdberg/}{M.~de~{Berg}}, \href{http://www.win.tue.nl/~ocheong}{O.~{Cheong}}, {M. van} Kreveld, and \href{http://www.cs.uu.nl/people/markov/}{M.~H. {Overmars}}.
\newblock  {\em Computational Geometry: Algorithms and Applications}.
\newblock Springer-Verlag, 3rd edition, 2008.

\bibitem[DRS09]{drs-pddd-09}
N.~N. Dalvi, C.~R{\'e}, and D.~Suciu.
\newblock  Probabilistic databases: Diamonds in the dirt.
\newblock {\em Commun. ACM}, 52(7):86--94, 2009.

\bibitem[DSST89]{dsst-mdsp-89}
J.~R. Driscoll, N.~Sarnak, D.~D. Sleator, and R.~E. Tarjan.
\newblock  Making data structures persistent.
\newblock {\em J. Comput. Syst. Sci.}, 38:86--124, 1989.

\bibitem[DYM{\etalchar{+}}05]{dymtv-psqeu-05}
X.~Dai, M.~L. Yiu, N.~Mamoulis, Y.~Tao, and M.~Vaitis.
\newblock  Probabilistic spatial queries on existentially uncertain data.
\newblock In {\em Proc. 9th Int. Sympos. Spatial Temporal Databases}, pages
  400--417, 2005.

\bibitem[{Har}11]{h-gaa-11}
\href{http://sarielhp.org}{S.~{{Har-Peled}}}.
\newblock  {\em Geometric Approximation Algorithms}, volume 173 of {\em
  Mathematical Surveys and Monographs}.
\newblock Amer. Math. Soc., 2011.

\bibitem[JCLY11]{jcly-srqpd-11}
J.~Jestes, G.~Cormode, F.~Li, and K.~Yi.
\newblock  Semantics of ranking queries for probabilistic data.
\newblock {\em IEEE Trans. Know. Data Eng.}, 23(12):1903--1917, 2011.

\bibitem[KCS14]{kcs-cppop-14}
P.~Kamousi, \href{http://www.math.uwaterloo.ca/~tmchan/}{T.~M.~{Chan}}, and S.~Suri.
\newblock  Closest pair and the post office problem for stochastic points.
\newblock {\em Comput. Geom. Theory Appl.}, 47(2):214--223, 2014.

\bibitem[KKR07]{kkr-pnnqu-07}
H.-P. Kriegel, P.~Kunath, and M.~Renz.
\newblock  Probabilistic nearest-neighbor query on uncertain objects.
\newblock In {\em Proc. 12th Int. Conf. Database Sys. Adv. App.}, pages
  337--348, 2007.

\bibitem[KMR{\etalchar{+}}16]{KMRSS}
\href{http://www.cs.tau.ac.il/~haimk/}{H.~{Kaplan}}, W.~Mulzer, L.~Roditty, P.~Seiferth, and \href{http://www.math.tau.ac.il/~michas}{M.~{Sharir}}.
\newblock  Dynamic planar voronoi diagrams for general distance functions and
  their algorithmic applications.
\newblock {\em CoRR}, abs/1604.03654, 2016.

\bibitem[LLS01]{lls-ibscl-01}
Y.~Li, P.~M. Long, and A.~Srinivasan.
\newblock  Improved bounds on the sample complexity of learning.
\newblock {\em J. Comput. Syst. Sci.}, 62(3):516--527, 2001.

\bibitem[LS07]{ls-aiap-07}
V.~Ljosa and A.~K. Singh.
\newblock  {APLA}: Indexing arbitrary probability distributions.
\newblock In {\em Proc. 23rd IEEE Int. Conf. Data Eng.}, pages 946--955, 2007.

\bibitem[MR95]{mr-ra-95}
R.~Motwani and P.~Raghavan.
\newblock  {\em Randomized Algorithms}.
\newblock Cambridge University Press, Cambridge, UK, 1995.

\bibitem[SA95]{sa-dsstg-95}
\href{http://www.math.tau.ac.il/~michas}{M.~{Sharir}} and \href{http://www.cs.duke.edu/~pankaj}{P.~K.~{Agarwal}}.
\newblock  {\em {Davenport-Schinzel} Sequences and Their Geometric
  Applications}.
\newblock Cambridge University Press, 1995.

\bibitem[SE08]{se-gvdus-08}
J.~Sember and W.~Evans.
\newblock  Guaranteed {Voronoi} diagrams of uncertain sites.
\newblock In {\em Proc. 20th Canad. Conf. Comput. Geom.}, pages 207--210, 2008.

\bibitem[VC71]{vc-ucrfe-71}
V.~N. Vapnik and A.~Y. Chervonenkis.
\newblock  On the uniform convergence of relative frequencies of events to
  their probabilities.
\newblock {\em Theory Probab. Appl.}, 16:264--280, 1971.

\bibitem[YTX{\etalchar{+}}10]{ytxpz-snnsu-10}
S.~M. Yuen, Y.~Tao, X.~Xiao, J.~Pei, and D.~Zhang.
\newblock  Superseding nearest neighbor search on uncertain spatial databases.
\newblock {\em IEEE Trans. Know. Data Eng.}, 22(7):1041--1055, 2010.

\bibitem[ZCM{\etalchar{+}}13]{zcmrz-vbnns-13}
P.~Zhang, R.~Cheng, N.~Mamoulis, M.~Renz, A.~Zufile, Y.~Tang, and T.~Emrich.
\newblock  Voronoi-based nearest neighbor search for multi-dimensional
  uncertain databases.
\newblock In {\em Proc. 29th IEEE Int. Conf. Data Eng.}, pages 158--169, 2013.

\end{thebibliography}

\end{document}